\pdfoutput=1
\documentclass[11pt]{article}

\usepackage[utf8]{inputenc}
\usepackage[T1]{fontenc}
\usepackage[english]{babel}
\usepackage{lmodern}
\usepackage{microtype}

\usepackage{amsmath, amssymb, amsthm, mathtools}
\usepackage{bm}

\usepackage{booktabs}
\usepackage{multirow}
\usepackage{algorithm}
\usepackage{algpseudocode}

\usepackage[margin=2.5cm]{geometry}
\usepackage{graphicx}
\graphicspath{{figures/}}
\usepackage{xcolor}
\definecolor{orcidgreen}{HTML}{A6CE39}

\usepackage[round]{natbib}

\usepackage[hidelinks]{hyperref}
\newtheorem{theorem}{Theorem}[section]
\newtheorem{proposition}[theorem]{Proposition}

\theoremstyle{definition}

\newcommand{\E}{\mathbb{E}}
\newcommand{\Var}{\mathrm{Var}}
\newcommand{\PATP}{\mathrm{PATP}}
\newcommand{\PMM}{\mathrm{PMM}}

\providecommand{\proglang}[1]{\textsf{#1}}

\title{Parametrically Adaptive Transition Polynomial:\\
a Signed-Parity Continuous-$\alpha$ Extension of Kunchenko Stochastic Polynomials}
\author{\textbf{Serhii V.\ Zabolotnii}\,\href{https://orcid.org/0000-0003-0242-2234}{\textsuperscript{\textcolor{orcidgreen}{\scriptsize ORCID}}}\\[0.5em]
\small Department of Information, Multimedia Technologies and Design,\\
\small Cherkasy State Business College, Cherkasy 18028, Ukraine\\[0.2em]
\small State Scientific Research Institute of Armament and Military Equipment\\
\small Testing and Certification, Cherkasy, Ukraine\\[0.2em]
\small Department of Cybernetics and Applied Mathematics,\\
\small Uzhhorod National University, Uzhhorod, Ukraine\\[0.4em]
\small \texttt{zabolotnii.serhii@csbc.edu.ua}}
\date{}

\begin{document}
\maketitle

\begin{abstract}
Kunchenko's method of polynomial maximization provides a semiparametric
apparatus for parameter estimation under non-Gaussian errors, but its
classical power basis relies on finite higher-order integer moments. This
paper introduces the Parametrically Adaptive Transition Polynomial (PATP),
a signed-parity fractional-power family controlled by a continuous
parameter $\alpha \in [0,1]$. The quadratic exponent map $p_i(\alpha)$
connects the fractal regime $p_i(0)=1/i$, the degenerate linear point
$p_i(1/2)=1$, and the signed-parity integer-power regime $p_i(1)=i$. For
the degree-$S=2$ case we derive a closed-form variance-reduction
coefficient $g_2(\alpha)$ in terms of signed and absolute fractional
moments, identify the singular behavior at $\alpha=1/2$, and state the
moment and regularity conditions under which the formula is meaningful.
The construction should be read as a Form-B PATP analogue within
Kunchenko's generalized apparatus, not as an exact recovery of the
canonical even-power PMM basis at $\alpha=1$. On symmetric canonical laws
(Laplace, generalized Gaussian, Student-$t$) a fully reproducible \proglang{R}
pipeline implementing the $\mathbf{F}_2^{-1}\mathbf{b}$ normal-equation solver
(Algorithm~1) shows the estimator's Monte Carlo efficiency converging to the
closed-form $g_2(\alpha)$, and confirms the same factor for regression
coefficients. On a panel of symmetric real series of increasing tail weight
(equity-index and exchange-rate log-returns) the data-driven optimal $\alpha$
tracks the tail shape, yielding variance reductions over the sample mean that
grow from $13\%$ to $54\%$ in agreement with theory; benchmarks against six
robust location baselines are included. Lean~4 verified algebraic facts support the
structural part of the derivation, and the Cauchy law marks the boundary of
applicability (infinite reference variance).
\end{abstract}


\medskip
\noindent\textbf{Code and data availability.} A fully reproducible R
pipeline and the Lean~4 formalization are publicly available at
\url{https://github.com/SZabolotnii/Ku-PATP-code-supplement} under an
MIT licence; running \texttt{Rscript R/run\_all.R} regenerates every
table and figure cited in this paper in under one minute on a single
thread.

\medskip
\noindent\textbf{Use of AI-assisted tools.} AI-assisted copy editing
(grammar, punctuation, and stylistic suggestions) was used during
manuscript preparation, and the AI coding assistant Claude Code
(Anthropic) was used to help write and debug the accompanying R
pipeline. All scientific content, derivations, and Lean proofs are the
author's own work, and the author has verified and takes full
responsibility for all code and numerical results.

\section{Introduction}\label{sec:intro}

\subsection{Kunchenko stochastic polynomials and the legacy of power bases}\label{sec:intro:kunchenko}

Kunchenko's theory of stochastic polynomials, developed by Yu.\,P.~Kunchenko
and systematized in the canonical monograph~\citep{kunchenko2002}, provides a
semiparametric language for working with non-Gaussian random processes through
a finite moment-cumulant description. Its applications include parameter
estimation by the method of polynomial maximization (PMM), hypothesis testing,
and pattern recognition in spaces with a generating element. In this broader
apparatus, PMM is the most important estimation branch, but the underlying
object is more general: a stochastic-polynomial representation whose behavior
depends crucially on the chosen basis.

For several decades the working basis of this school has effectively been the
classical integer power basis. This choice is natural: powers of the residual
connect the apparatus directly with central moments, cumulants, correlant
matrices, and the closed-form efficiency formulas that make PMM interpretable.
It also places Kunchenko's construction between ordinary least squares, which
uses essentially second-order information, and maximum likelihood, which
requires a full distributional model. For near-Gaussian variables whose
departure from normality is well summarized by finitely many higher-order
cumulants, this power-basis tradition is highly effective
~\citep{kunchenko2002,kunchenko2006stochastic}.

The same tradition, however, creates a structural limitation: power bases
inherit the moment requirements of the powers used in the stochastic
polynomial. The familiar second-degree PMM efficiency formula, recalled in
\S\ref{sec:bg:integer}, already depends on skewness and excess kurtosis. For
heavy-tailed distributions these quantities may be undefined or empirically
unstable, and the matrix of centered correlants in Kunchenko's construction
may cease to be meaningful.

This issue is not only formal. Cauchy and stable laws, high-variance
log-normal models, and empirical distributions of financial returns or
telecommunications loads routinely exhibit tail behavior for which third- and
fourth-order summaries are unreliable
~\citep{samorodnitsky1994stable,nolan2020stable,shao1993signal}.
Thus the classical power-basis version of the apparatus can lose its practical
advantage exactly in the regime where a semiparametric method is most
desirable: non-Gaussian data with too little reliable moment information for
standard cumulant-based calibration.

Kunchenko's monographs outline three applied branches of the apparatus:
statistical parameter estimation, hypothesis testing, and pattern recognition
~\citep{kunchenko2002,kunchenko2006stochastic}. The closest neighboring
Western traditions are GMM~\citep{hansen1982gmm}, Huber's M-estimators
~\citep{huber1981robust}, Hosking's L-moments~\citep{hosking1990lmoments},
and SLS~\citep{wang2008sls}.

\subsection{Fractional calculus and fractional-power bases}\label{sec:intro:fractional}

An alternative is offered by the apparatus of \emph{fractional calculus},
which generalizes the classical notions of differentiation and integration to
non-integer orders. In the
statistical context, the key concept is that of \emph{fractional absolute
moments}: expectations of non-integer powers of the centered absolute
residual. Such moments can remain finite for distributions whose integer
moments of order two or higher are absent. For Cauchy, for example,
fractional moments exist only below order one; for stable laws, the admissible
orders are bounded by the stability index. This is the reason fractional
lower-order moment methods are useful in heavy-tail signal
processing~\citep{matsui2013fracmoments,nolan2020stable}.

A natural continuation of this idea is the construction of
\emph{fractional-power bases} in the stochastic-polynomial construction:
instead of only integer powers, one uses sign-preserving fractional powers of
the residual.
When $p_i = i$ we obtain a signed-parity integer-power basis: it coincides
with the classical power basis for odd powers, while even powers remain in
the sign-preserving Form-B convention. When $p_i = 1/i$ we obtain the
\emph{fractal} basis, which operates with fractional absolute moments. Its
applicability is still governed by
the specific moment orders required by the estimator; extremely heavy-tailed
laws such as Cauchy are therefore limiting cases for the $S=2$ construction.
The conceptual prerequisites for such an extension
come from two established lines: fractional lower-order moments in
heavy-tailed signal processing~\citep{shao1993signal,matsui2013fracmoments}
and Kunchenko's generating-element approximation space
~\citep{kunchenko2005poligen}.

\subsection{The Parametrically Adaptive Transition Polynomial (PATP)}\label{sec:intro:patp}

The construction proposed in this paper --- the \emph{Parametrically Adaptive
Transition Polynomial} (PATP) --- singles out the fractional-power
basis as a \emph{continuously parametrized} family within Kunchenko's
apparatus. We introduce a control parameter $\alpha \in [0, 1]$ and use it
to move continuously between three structural regimes of the basis. The
formal exponent map and the exact basis definition are given later in
\S\ref{sec:patp:param}; at the level of the introduction, the essential point
is that the chosen quadratic parameterization realizes the following three
anchor states:
\begin{itemize}
  \item $\alpha = 0$: $p_i = 1/i$ --- the \emph{fractal} regime, suitable
        for heavy tails;
  \item $\alpha = \tfrac{1}{2}$: $p_i = 1$ for all $i$ --- the degenerate
        \emph{linear} regime;
  \item $\alpha = 1$: $p_i = i$ --- the signed-parity
        \emph{integer-power} regime related to the classical
        power-polynomial basis (\citealp{kunchenko2002}, ch.~4).
\end{itemize}
This paper uses that idea to ``smoothly vary the type of polynomial basis''
according to the estimated shape of the tails of the residual distribution.
We formalize this construction within Kunchenko's apparatus, derive a
closed-form formula for the variance reduction coefficient $g_2(\alpha)$ for
$S = 2$, identify the degeneracy at $\alpha = 1/2$, and examine the behavior
of the estimator on canonical distributions.

PATP does not replace the classical apparatus. In the Form-B convention used
here, $\alpha = 1$ gives a signed-parity analogue of the integer-power basis
rather than an exact recovery of the canonical even-power PMM basis.
Symmetrically, the fractal regime $\alpha = 0$ reduces the required integer
moment order, but it does not cover all heavy-tailed laws automatically:
Cauchy remains outside the $S=2$ OLS-referenced variance formula because the
second moment and the required $\nu_{3/2}$ moment are infinite. Between
these two poles, PATP provides a continuous ``dial'' for basis selection,
optimized to the estimated non-Gaussianity profile of the noise when the
stated moment conditions hold.

\subsection{Contributions and structure}\label{sec:intro:contributions}

The contributions of this paper are:

\begin{enumerate}
  \item \textbf{Formal definition of PATP} as a continuously
        $\alpha$-parametrized family of basis functions within the
        generalized apparatus of Kunchenko stochastic polynomials
        (\S\ref{sec:patp:family}--\S\ref{sec:patp:criterion}); the structural
        degeneration at $\alpha = \tfrac{1}{2}$ is identified, and
        non-degeneracy conditions for the polynomial body
        $\Delta_S(\alpha)$ are stated.

  \item \textbf{Closed-form formula $g_2(\alpha; \nu_p, \gamma_3)$} for
        the case $S = 2$ (\S\ref{sec:eff:s2}), expressed through fractional
        absolute moments $\nu_p = \E[|\xi|^p]$. At $\alpha = 1$ the formula
        yields the Form-B signed-parity analogue of the classical PMM2
        expression rather than the canonical even-power result; at
        $\alpha = 0$ it is expressed through fractional moments including
        $\nu_{3/2}$ and can remain meaningful for finite-variance laws even
        when the fourth moment and $\gamma_4$ are unavailable.

  \item \textbf{Degeneracy and numerical shape analysis} of $g_2(\alpha)$
        (\S\ref{sec:eff:unimodal}); the structural point
        $\alpha = \tfrac{1}{2}$ is identified as singular for the
        $\mathbf{F}_2$ matrix, while typical curve shapes are examined
        numerically rather than claimed as a general theorem.

  \item \textbf{Implementation algorithms} (\S\ref{sec:algorithms}):
        a protected fixed-$\alpha$ solver stack and several calibration
        variants for selecting $\alpha^*$ from fractional-moment and robust
        shape diagnostics.

  \item \textbf{Reproducible validation and a real-data application}
        (\S\ref{sec:illustrations}): on symmetric canonical laws (Laplace,
        Generalized Gaussian, Student-$t$; Cauchy as a limiting case) the full
        $\mathbf{F}_2^{-1}\vec b$ estimator's Monte Carlo $g_2(\alpha)$ converges
        to the closed-form prediction, the same factor is confirmed for
        regression coefficients, and a real-data location example (daily index
        log-returns) attains a $13\%$ variance reduction over the sample mean,
        competitive with classical robust baselines.

  \item \textbf{Conceptual placement} of PATP among neighboring
        semiparametric frameworks --- GMM, M-estimators, L-moments,
        SLS (\S\ref{sec:disc:related}) --- and an honest account of
        limitations (\S\ref{sec:disc:robust}).
\end{enumerate}

Empirical cross-domain validation and the applied implementation of the
$\alpha = 0$ fractal case in Royston-Altman fractional polynomial
regression~\citep{royston1994} are left to separate work.

The structure of the remainder of the paper is as follows. In
\S\ref{sec:background} we recall Kunchenko's generalized apparatus, derive
the closed-form formula $g_2$ for the power basis, and introduce fractional
absolute moments. In \S\ref{sec:patp} we formally define PATP. In
\S\ref{sec:efficiency} we develop the theoretical efficiency analysis,
including closed-form formulas for $S = 2$. In \S\ref{sec:algorithms} we
describe the algorithms. In \S\ref{sec:illustrations} we present
illustrations. \S\ref{sec:discussion} discusses connections with neighboring
frameworks, and \S\ref{sec:conclusion} summarizes the paper.

\section{Background}\label{sec:background}

This section unfolds the three layers of apparatus on which the PATP
construction is built, and positions it in the context of modern methods of
robust statistics and signal processing.
\S\ref{sec:bg:literature} provides a literature review on the use of
fractional polynomials, semiparametric methods, and estimation approaches
under heavy-tailed distributions.
\S\ref{sec:bg:generalized} recalls the definition of the \emph{generalized}
Kunchenko stochastic polynomial and the maximum property on which PMM is
based. \S\ref{sec:bg:integer} specializes the apparatus to the classical
power basis and derives the closed-form formula
$g_2 = 1 - \gamma_3^2/(2+\gamma_4)$ for $S=2$ via the system of normal
equations. \S\ref{sec:bg:fracmoments} introduces fractional absolute moments
$\nu_p = \E[|\xi|^p]$ and fixes the conceptual place of the fractional-power
basis in Kunchenko's approximation space with a generating
element~\citep{kunchenko2005poligen}. All notation used in~\S\ref{sec:patp}
is introduced here for the first time.

\subsection{Literature review: fractional polynomials, heavy tails, and semiparametric methods}\label{sec:bg:literature}

The modern development of statistical estimation methods demonstrates a
sustained departure from the paradigm of exclusively classical polynomial and
Gaussian approximations toward more flexible and robust semiparametric
constructions. In this context, the Parametrically Adaptive Transition
Polynomial (PATP) inherits and generalizes ideas from several adjacent
fields simultaneously: biostatistics (fractional polynomials), robust signal
processing (lower-order statistics), and heavy-tail theory (stable laws).

\paragraph{Fractional polynomials and flexible modeling.}
The methodology of using non-integer (fractional) or negative powers in
regression analysis was fundamentally established in the work of Royston and
Altman~\citep{royston1994}. This paradigm proved its superiority over
classical polynomial approximation (which often suffers from the Runge
phenomenon and uncontrolled oscillations at the boundaries of the sample),
since it provides a wider spectrum of functional forms with fewer parameters
(parsimony parameter model). The modern stage of development of this
apparatus is consolidated in the monograph on mixed multivariable
modeling~\citep{royston2008multivariable, sauerbrei1999building}, where
fractional polynomials (FP) serve as the standard for finding nonlinear
relationships of unknown form. The advantages of fractional polynomials in
the context of robustness of estimates (in particular, dosing and risk
analysis) were vividly demonstrated by Faes et al.~\citep{faes2003fractional},
who showed their ability to minimize the influence of model specification
errors. Direct adaptation of the fractional-power trend
($p_i \in \{0.5, -1, \dots\}$) motivates the PATP basis construction
developed in this paper.

\paragraph{Heavy-tailed distributions and robust signal processing.}
The limitations of the classical (power) method of moments and the method of
least squares (OLS) become most fatal in the case of heavy-tailed
distributions, where higher moments or even the variance may be infinite. The
canonical theory of such $\alpha$-stable non-Gaussian processes is described
in the monograph of Samorodnitsky and Taqqu~\citep{samorodnitsky1994stable}.
In practice, this problem is most frequently encountered in radar, telecommunications,
and financial market analysis~\citep{nolan2020stable, matsui2013fracmoments}.
To overcome the divergence of higher-order moments, modern signal processing
methods use \emph{fractional lower order moments} (FLOM) --- moments of
fractional powers lower than the stability order. Thus, the classical work of
Shao and Nikias~\citep{shao1993signal} directly justified the use of
fractional absolute moments for the synthesis of optimal filters and receivers
in environments with $\alpha$-stable impulsive noise, which conceptually is
the primary basis of the fractal (limiting at $\alpha = 0$) regime of the
PATP construction. The general principles of robustness under such
interference are consolidated in modern handbooks on robust signal
processing~\citep{zoubir2012robust}.

\paragraph{Semiparametric procedures and L-moments.}
The development of PATP should also be understood as an evolution of
estimation approaches under limited prior information without assumptions
about the availability of the full distribution function. These classical
ideas, laid down in Hansen's generalized method of moments
(GMM)~\citep{hansen1982gmm} and Huber's robust statistics~\citep{huber1981robust},
find further continuation in Hosking's concept of
L-moments~\citep{hosking1990lmoments}. Separately for tails, Trimmed
L-moments were created~\citep{elamir2003trimmed}, which stabilize estimates
by trimming extreme order statistics. The introduction of $p_i(\alpha)$ in
PATP (\S\ref{sec:patp:param}) plays an analogous role, but achieves
robustness not through threshold rejection (or trimming), but through
continuous smoothing of the nonlinearity of the basis to a ``safe'' range of
fractional exponents.

\paragraph{Modern robust mean estimation.}
Recent heavy-tail theory also includes estimators whose primary objective is
non-asymptotic concentration under weak moment assumptions. Catoni-type
estimators~\citep{catoni2012challenging}, median-of-means and geometric-median
aggregation~\citep{minsker2015geometric}, and the broader survey of
heavy-tailed mean and regression methods by Lugosi and
Mendelson~\citep{lugosi2019survey} provide the closest modern robust
benchmark family. PATP differs in emphasis: it is not a distribution-free
concentration method, but a Kunchenko-style fractional-moment construction
with an explicit efficiency coefficient when the required signed and absolute
moments exist.

Thus, PATP combines the flexibility of Royston-Altman FP modeling with the
FLOM ideas of heavy-tail-robust signal processing, embedding them in the
rigorous mathematical apparatus of Kunchenko stochastic polynomial
maximization.

\subsection{The generalized Kunchenko stochastic polynomial}\label{sec:bg:generalized}

\paragraph{Setup.}
Let $\xi$~be an observed random variable whose distribution depends on an
unknown scalar parameter $\theta \in \Theta \subset \mathbb{R}$.
Consider a set of basis functions $\{\varphi_i(\cdot)\}_{i=1}^{S}$,
linearly independent on the domain of admissible values of~$\xi$, with
properties
\begin{equation}\label{eq:psi-def}
  \Psi_{ij}(\theta) \;:=\; \E\bigl[\varphi_i(\xi)\,\varphi_j(\xi)\bigr],
  \qquad
  \Psi_i(\theta) \;:=\; \E\bigl[\varphi_i(\xi)\bigr],
\end{equation}
finite for $i, j = 1, \ldots, S$. As shown in Kunchenko's monographs
~\citep{kunchenko2005poligen,kunchenko2006stochastic}, known
subclasses of bases include:

\begin{itemize}
  \item \emph{power class} $\varphi_i(\xi) = \xi^i$, for which
        $\Psi_{ij} = \alpha_{i+j}$ and $\Psi_i = \alpha_i$ via the initial
        moments $\alpha_k = \E[\xi^k]$;
  \item \emph{trigonometric class} $\varphi_i(\xi) = \cos(ik\xi)$ or
        $\sin(ik\xi)$, expressed through the characteristic function;
  \item \emph{exponential class} $\varphi_i(\xi) = e^{ik\xi}$.
\end{itemize}
In the present paper we work only with the power class and its
fractional-power extension (\S\ref{sec:patp:family}); the trigonometric
and exponential subclasses remain out of scope.

\paragraph{Stochastic polynomial and the matrix of centered correlants.}
Kunchenko~\citep[ch.~3]{kunchenko2002} calls the random variable
\begin{equation}\label{eq:eta-S}
  \eta_S(\xi)
  \;=\; h_0 \;+\; \sum_{i=1}^{S} h_i\,\varphi_i(\xi),
  \qquad |h_i| < \infty,
\end{equation}
the \emph{generalized stochastic polynomial} of degree~$S$ in the basis
$\{\varphi_i\}$, where $h_0, \ldots, h_S$ are non-random coefficients.
The key characteristic of such a polynomial is the \emph{matrix of centered
correlants}
\begin{equation}\label{eq:F-matrix}
  F_{ij}(\theta)
  \;:=\; \Psi_{ij}(\theta) \;-\; \Psi_i(\theta)\,\Psi_j(\theta),
  \qquad i, j = 1, \ldots, S.
\end{equation}
The symmetric matrix $\mathbf{F}_S(\theta) = \bigl(F_{ij}(\theta)\bigr)$ is
a non-trivial generalization of the covariance matrix to higher moments; its
determinant $\Delta_S(\theta) = |\mathbf{F}_S(\theta)|$ is called the
\emph{body} of the stochastic polynomial by
Kunchenko~\citep{kunchenko2006stochastic}. Under linear independence of
$\{\varphi_i\}_{i=1}^{S}$, the matrix $\mathbf{F}_S$ is positive definite,
i.e., $\Delta_S > 0$, which guarantees the existence and uniqueness of the
solution to the estimation problem.

The variance of the polynomial~\eqref{eq:eta-S} itself is expressed by the
quadratic form
\begin{equation}\label{eq:var-eta}
  \Var\bigl[\eta_S(\xi)\bigr]
  \;=\; \sum_{i=1}^{S}\sum_{j=1}^{S} h_i\,h_j\,F_{ij}(\theta).
\end{equation}

\paragraph{Global maximum property and the PMM criterion.}
The method of polynomial maximization is based on the property proved by
Kunchenko~\citep[ch.~3.2]{kunchenko2002}: under certain regularity conditions
on the coefficients $h_i$, the expected value of the stochastic polynomial,
viewed as a function of the parameter~$\theta$, has a \emph{global maximum
in a neighborhood of the true value}~$\theta_0$:
\begin{equation}\label{eq:mmpl-property}
  \theta_0
  \;=\; \arg\max_{\theta \in \Theta}\;
        \E\bigl[\eta_S(\xi;\theta)\bigr].
\end{equation}
The PMM estimate of degree~$S$ from a sample $x_1, \ldots, x_N$ is defined
as the point of maximum of the \emph{empirical} analog
\begin{equation}\label{eq:LS-empirical}
  \hat\theta_{\PMM,S}
  \;=\; \arg\max_{\theta \in \Theta}\;
        L_S(\theta; x_1, \ldots, x_N),
  \qquad
  L_S(\theta; x_1, \ldots, x_N)
  \;=\; \sum_{i=1}^{S} h_i^{*}(\theta) \cdot
        \frac{1}{N}\sum_{n=1}^{N} \varphi_i(x_n),
\end{equation}
where $h_i^{*}(\theta)$ are the \emph{optimal} coefficients that minimize
the asymptotic variance of the estimate. The solution to this minimization
problem~\citep[ch.~3.3]{kunchenko2002} is the linear system
\begin{equation}\label{eq:normal-eq}
  \mathbf{F}_S(\theta) \cdot \vec h^{*}(\theta)
  \;=\; \vec b(\theta),
\end{equation}
where $\vec b(\theta) \in \mathbb{R}^{S}$ is the vector with components
$b_i(\theta) = \partial \Psi_i / \partial \theta$, which depend on the
derivatives of the basis characteristics with respect to the parameter.

\paragraph{Variance of the estimate and the variance reduction coefficient.}
Substituting $\vec h^{*}(\theta) = \mathbf{F}_S^{-1}(\theta)\,\vec b(\theta)$
into the functional~\eqref{eq:LS-empirical} and taking the second derivative
with respect to the parameter yields the canonical formula for the asymptotic
variance of the PMM estimate~\citep[ch.~3.4]{kunchenko2002}
\begin{equation}\label{eq:var-pmm-canonical}
  \Var\bigl[\hat\theta_{\PMM,S}\bigr]
  \;=\; \frac{1}{N \cdot
                \vec b^{\top}(\theta_0)\,
                \mathbf{F}_S^{-1}(\theta_0)\,
                \vec b(\theta_0)}.
\end{equation}
The ratio of this variance to the variance of the OLS estimate gives the
\emph{variance reduction coefficient}
\begin{equation}\label{eq:gS-def}
  g_S(\theta)
  \;:=\; \frac{\Var\!\bigl[\hat\theta_{\PMM,S}\bigr]}
              {\Var\!\bigl[\hat\theta_{\mathrm{OLS}}\bigr]},
  \qquad 0 < g_S(\theta) \leq 1,
\end{equation}
which is interpreted as the asymptotic relative efficiency (ARE) of PMM
against OLS. The value $g_S = 1$ corresponds to no gain (Gaussian case);
$g_S < 1$ is a strict gain of PMM in variance. As $S$ increases, the
coefficient $g_S$ decreases
monotonically~\citep[ch.~3.5]{kunchenko2002}, asymptotically approaching
the Cramér-Rao efficiency bound: as $S \to \infty$ the PMM estimate is
asymptotically equivalent to the maximum likelihood estimate, while not
requiring knowledge of the full distribution function.

\subsection[The power class and derivation of g2 for S = 2]{The power class and derivation of $g_2$ for $S = 2$}\label{sec:bg:integer}

Consider the most important special case: estimation of the mean
$\theta = \mu$ of a random variable $\xi$ with finite moments up to the
fourth order in the power basis $\varphi_i(\xi) = (\xi - \mu)^i$,
$i = 1, 2$.

\paragraph{Parametrization via central moments.}
We introduce notation for central moments and cumulants:
\begin{equation}\label{eq:moments-cumulants}
  m_k(\mu) := \E\bigl[(\xi - \mu)^k\bigr], \qquad
  c_2 = m_2, \quad c_3 = m_3, \quad c_4 = m_4 - 3 m_2^2,
\end{equation}
where $c_k$~is the $k$-th cumulant, and the relation $m_4 = c_4 + 3 c_2^2$
is classical. Standardized cumulants:
$\gamma_3 = c_3 / c_2^{3/2}$ (skewness coefficient) and
$\gamma_4 = c_4 / c_2^2$ (excess kurtosis).

\paragraph{Matrix $\mathbf{F}_2$ and vector $\vec b$.}
At the true value $\mu$ we have $\Psi_i(\mu) = \E[(\xi - \mu)^i] = m_i$,
and since $m_1 = 0$:
\begin{align}
  F_{11} &= \E[(\xi - \mu)^2] - \E[\xi - \mu]^2
          \;=\; m_2 \;=\; c_2,
          \label{eq:F11-int}\\
  F_{12} = F_{21} &= \E[(\xi - \mu)^3] - \E[\xi - \mu]\E[(\xi - \mu)^2]
          \;=\; m_3 \;=\; c_3,
          \label{eq:F12-int}\\
  F_{22} &= \E[(\xi - \mu)^4] - \E[(\xi - \mu)^2]^2
          \;=\; m_4 - m_2^2
          \;=\; c_4 + 2c_2^2.
          \label{eq:F22-int}
\end{align}
Hence
\begin{equation}\label{eq:F2-matrix}
  \mathbf{F}_2(\mu)
  \;=\;
  \begin{pmatrix}
    c_2 & c_3 \\
    c_3 & c_4 + 2c_2^2
  \end{pmatrix},
  \qquad
  \Delta_2 \;=\; c_2(c_4 + 2c_2^2) - c_3^2.
\end{equation}

The vector $\vec b$ is computed from $b_i = \partial \Psi_i / \partial \mu
= -i \cdot \E[(\xi-\mu)^{i-1}] = -i \cdot m_{i-1}$. With $m_0 = 1$ and
$m_1 = 0$:
\begin{equation}\label{eq:b2-vector}
  \vec b
  \;=\;
  \begin{pmatrix} -1 \\ 0 \end{pmatrix}.
\end{equation}
The sign $-1$ can be absorbed into a redefinition of the coefficients $h_i$
without affecting the variance of the estimate; hereafter we work with
$\vec b = (1, 0)^{\top}$.

\paragraph{Solution of the system and the coefficient $g_2$.}
Substituting~\eqref{eq:F2-matrix} and $\vec b = (1, 0)^{\top}$
into~\eqref{eq:var-pmm-canonical}:
\begin{equation}\label{eq:b-Finv-b}
  \vec b^{\top}\,\mathbf{F}_2^{-1}\,\vec b
  \;=\; \bigl(\mathbf{F}_2^{-1}\bigr)_{11}
  \;=\; \frac{F_{22}}{\Delta_2}
  \;=\; \frac{c_4 + 2c_2^2}{c_2(c_4 + 2c_2^2) - c_3^2}.
\end{equation}
Therefore,
\begin{equation}\label{eq:var-pmm2}
  \Var\bigl[\hat\mu_{\PMM,2}\bigr]
  \;=\; \frac{1}{N}\cdot\frac{c_2(c_4 + 2c_2^2) - c_3^2}{c_4 + 2c_2^2}
  \;=\; \frac{c_2}{N}
        \left[1 - \frac{c_3^2}{c_2(c_4 + 2c_2^2)}\right].
\end{equation}
Since $\Var[\hat\mu_{\mathrm{OLS}}] = c_2/N$, the variance reduction
coefficient takes the closed form
\begin{equation}\label{eq:g2-final}
  g_2 \;=\; 1 \;-\; \frac{c_3^2}{c_2(c_4 + 2c_2^2)}.
\end{equation}
Via the standardized cumulants $\gamma_3 = c_3 / c_2^{3/2}$ and
$\gamma_4 = c_4/c_2^2$ we have
$c_3^2/[c_2(c_4 + 2c_2^2)] = \gamma_3^2 / (\gamma_4 + 2)$, which gives
the canonical result of Kunchenko~\citep[formula~2.18]{kunchenko2002}:
\begin{equation}\label{eq:g2-standardized}
  \boxed{\,g_2 \;=\; 1 \,-\, \frac{\gamma_3^2}{2 + \gamma_4}.\,}
\end{equation}

\paragraph{Interpretation.}
Formula~\eqref{eq:g2-standardized} reveals two structural aspects of
PMM-2 efficiency.
First, the gain over OLS is \emph{proportional to the square of the
skewness} $\gamma_3^2$: for symmetric distributions ($\gamma_3 = 0$) there
is no gain. Second, the kurtosis $\gamma_4$ enters as a ``discount factor''
in the denominator: heavy tails (large $\gamma_4 > 0$) reduce the gain,
while light tails ($\gamma_4 < 0$, bounded by the condition
$\gamma_4 > -2$) amplify it. For the Gaussian case
$\gamma_3 = \gamma_4 = 0$, accordingly $g_2 = 1$ and PMM coincides with
OLS. The derivation for $S = 3$ from symmetric distributions analogously
gives $g_3 = 1 - \gamma_4^2/(6 + 9\gamma_4 + \gamma_6)$
\citep{zabolotnii2019symmetric}.

\subsection{Fractional calculus and fractional absolute moments}\label{sec:bg:fracmoments}

\paragraph{Limitations of the power basis.}
The derivation of~\eqref{eq:g2-standardized} relies on the finiteness of
moments $\alpha_k = \E[\xi^k]$ at least up to $k = 2S = 4$. In
\S\ref{sec:intro:kunchenko} we already noted that this condition is violated
for a wide class of heavy-tailed distributions: Cauchy ($\E[|\xi|] = \infty$),
Lévy stable with $\alpha_s < 2$ ($\Var[\xi] = \infty$), log-normal with
large variance. In such cases, the very construction of the matrix
$\mathbf{F}_2$ via~\eqref{eq:F11-int}--\eqref{eq:F22-int} becomes formally
undefined.

\paragraph{Fractional absolute moments.}
A natural candidate replacement for the integral moments $\alpha_k$ is the
\emph{fractional absolute moments}
\begin{equation}\label{eq:fracmoment-def}
  \nu_p \;:=\; \E\bigl[|\xi|^p\bigr], \qquad p \in (0, +\infty),
\end{equation}
which retain finiteness for a significantly wider class of distributions.
For Lévy stable laws with stability index $\alpha_s$, the moments
$\nu_p < \infty$ everywhere that $p < \alpha_s$~\citep{nolan2020stable};
for the Cauchy distribution ($\alpha_s = 1$) we have $\nu_p < \infty$ for
$p < 1$. Fractional lower-order moment methods use this fact to build
estimators and filters in impulsive-noise settings
~\citep{shao1993signal,matsui2013fracmoments}. For distributions with finite
variance, $\nu_2 \equiv c_2$ and the ordinary second-moment description is
recovered.

\paragraph{Fractional-power basis functions.}
Fractional absolute moments~\eqref{eq:fracmoment-def} are the natural
numerical characteristics of \emph{fractional-power} basis functions
\begin{equation}\label{eq:gp-def}
  g_p^{+}(\xi) \;:=\; |\xi|^{p}, \qquad
  g_p^{-}(\xi) \;:=\; \mathrm{sign}(\xi) \cdot |\xi|^{p},
  \qquad p \in \mathbb{R}_{+},
\end{equation}
which generalize $\xi^i$ ($i \in \mathbb{N}$) to non-integer exponents
while preserving parity (sign symmetry in~$g_p^{-}$, symmetry in~$g_p^{+}$).
For centered~$\xi$ we have $\E[g_p^{+}(\xi)] = \nu_p$ and
$\E[g_p^{-}(\xi)] = \nu_p^{(s)}$, where
$\nu_p^{(s)} := \E[\mathrm{sign}(\xi)\,|\xi|^p]$ is the signed analog of
the fractional moment, which vanishes for symmetric distributions.

\paragraph{Kunchenko's generating element.}
The conceptual place of the fractional-power extension in Kunchenko's school
is determined through the apparatus of the \emph{space with a generating
element}~\citep{kunchenko2005poligen}. In this
formalism, the stochastic polynomial exists as an abstract object in a
functional space, independently of the specific class of basis functions.
The power, trigonometric, and exponential classes from
\S\ref{sec:bg:generalized} are different \emph{realizations} of this
abstract polynomial in different basis spaces. The fractional-power functions
$g_{p}^{\pm}$ from~\eqref{eq:gp-def} form yet another realization,
conceptually no less canonical than the power class, but with an extended
range of applicability in terms of the moment characteristics of the noise.

It is precisely this conceptual position --- the fractional-power class as an
equal realization of the generalized stochastic polynomial --- that serves as
the structural justification for introducing the PATP construction in
\S\ref{sec:patp}: the parametrization $\alpha \in [0, 1]$ provides a
\emph{continuous path} between two realizations (the power one at
$\alpha = 1$ and the fractal one at $\alpha = 0$) within a single apparatus.

\subsection{Contrexcess and the entropy coefficient: topographic classification}\label{sec:bg:entropy}

Classification of non-Gaussian distributions by the kurtosis coefficient
$\gamma_4$ alone is fundamentally incomplete: structurally different
distributions may exist with identically close values of~$\gamma_4$, yet
substantially different tail and central-region shapes. In particular, for
symmetric distributions ($\gamma_3 = 0$), where the skewness parameter is
uninformative, this problem is exacerbated: the pair (Laplace, $\gamma_4 = 3$)
and (Simpson-triangular, $\gamma_4 = -0.6$) in the current PATP calibration
$\alpha^{*}(\gamma_3, \gamma_4)$ appear as different classes, while their
optimal $\alpha^{*}$ may be close. This section introduces a second
independent shape parameter --- the \emph{entropy coefficient}~$k$ --- and
the corresponding \emph{topographic classification} on the plane
$(\varkappa, k)$, historically proposed in the metrological school of
P.\,V.~Novitsky and Z.\,L.~Zograf~\citep{novitsky1991errors}.

\paragraph{Differential entropy and the entropic error value.}
The \emph{Shannon differential entropy} of an absolutely continuous random
variable $\xi$ with density $f_{\xi}(\cdot)$:
\begin{equation}\label{eq:shannon-entropy}
  H(\xi) \;:=\; -\int_{\mathbb{R}} f_{\xi}(x)\,\ln f_{\xi}(x)\,dx,
\end{equation}
defined under the condition of absolute integrability of the function
$f \ln f$. In the metrological tradition of Novitsky it is more convenient
to work with the \emph{entropic error value}~\citep[ch.~3.4]{novitsky1991errors}:
\begin{equation}\label{eq:entropic-error}
  \Delta_{\text{e}}(\xi) \;:=\; \frac{1}{2}\,e^{H(\xi)}.
\end{equation}
The interpretation of~$\Delta_{\text{e}}$ is the \emph{half-width} of the
uniform distribution that has the same differential entropy as~$\xi$ (i.e.,
equivalent ``uncertainty'' by the Shannon criterion).

\paragraph{Entropy coefficient.}
The \emph{entropy coefficient} of distribution~$\xi$ is defined as the ratio
\begin{equation}\label{eq:entropy-coeff}
  k(\xi)
  \;:=\; \frac{\Delta_{\text{e}}(\xi)}{\sigma(\xi)}
  \;=\; \frac{e^{H(\xi)}}{2\,\sigma(\xi)},
\end{equation}
where $\sigma(\xi) = \sqrt{\Var[\xi]}$. This is a dimensionless quantity
(invariant to scale and shift), completely characterized by the shape of the
distribution law. The alternative form of~\eqref{eq:entropy-coeff}, expressed
directly through the Shannon entropy, is convenient for theoretical analysis;
the form through $\Delta_{\text{e}}$ emphasizes the metrological
interpretation of~$k$ as the ``entropic width'' normalized to the standard
deviation.

\paragraph{Properties of $k$.}
\begin{itemize}
  \item \emph{Scale and shift invariance:}
        $k(c\xi + b) = k(\xi)$ for all $c \neq 0$, $b \in \mathbb{R}$
        (via the properties of entropy: $H(c\xi+b) = H(\xi) + \ln|c|$ and
        $\sigma(c\xi+b) = |c|\sigma(\xi)$).
  \item \emph{Positivity:} $k(\xi) > 0$ for all distributions with finite
        entropy and finite variance.
  \item \emph{Upper bound:} $k(\xi) \leq k_{\max} = \sqrt{2\pi e}/2 \approx 2.0663$,
        with equality only for the Gaussian distribution. This is a
        consequence of the \emph{maximum entropy
        principle}~\citep[ch.~12.4]{coverthomas2006}: among all distributions
        with fixed variance, the Gaussian maximizes entropy.
  \item \emph{Lower bound:} $k(\xi) > 0$ strictly (formally there is no
        universal lower bound, but for absolutely continuous distributions
        with bounded density $\sup f < \infty$ we have
        $k \gtrsim 1/(2\sigma \sup f)$).
\end{itemize}

\paragraph{Table of $k$ for canonical distributions.}
\begin{center}
\begin{tabular}{l|c|c|c}
Distribution & $\gamma_4$ & $\varkappa = 1/\sqrt{\gamma_4 + 3}$ & $k$ \\
\hline
Gaussian            & $0$      & $0.577$ & $2.0663$ \\
Laplace             & $3$      & $0.408$ & $1.9300$ \\
Simpson (triangular)& $-0.6$   & $0.645$ & $2.0240$ \\
Uniform             & $-1.2$   & $0.745$ & $1.7321$ \\
Arcsine             & $-1.5$   & $0.816$ & $1.1107$ \\
Cauchy              & $\infty$ & $0$     & undefined$^{*}$  \\
\end{tabular}
\end{center}
The asterisk ($^{*}$) marks the degenerate case: for Cauchy the entropy is
finite ($H = \ln(4\pi)$), but the variance is infinite, so $k$ is not a
well-defined coefficient. It is kept in the table only as a boundary marker
for regimes outside the finite-variance topographic plane.

\paragraph{Topographic plane $(\varkappa, k)$.}
The \emph{contrexcess} of a symmetric distribution is defined as
\begin{equation}\label{eq:contrexcess}
  \varkappa \;:=\; \frac{1}{\sqrt{\gamma_4 + 3}}
            \;=\; \frac{\sigma^2}{\sqrt{m_4}},
\end{equation}
which normalizes the kurtosis to the interval $\varkappa \in (0, 1]$ (under
the condition $\gamma_4 \geq -2$; for $\gamma_4 = -2$ we have $\varkappa = 1$
as the degenerate two-point distribution; for $\gamma_4 \to \infty$~--- we
have $\varkappa \to 0$, heavy tails).

The \emph{topographic plane} of the Novitsky-Zograf school is the coordinate
plane $(\varkappa, k)$, on which each symmetric distribution is represented
by a unique point, and each parametric family by a
curve~\citep{novitsky1991errors}. An example of such a family is the
Generalized Gaussian Distribution~$\mathrm{GG}(\beta)$: at $\beta = 1$ we
have Laplace, at $\beta = 2$~--- Gaussian, at $\beta \to \infty$~---
Uniform. On the plane $(\varkappa, k)$ this family forms a continuous curve
passing through the three corresponding points of the table.
Figure~\ref{fig:topographic} visualizes the canonical distributions and the
GG family.

\begin{figure}[!htb]
  \centering
  \includegraphics[width=0.85\linewidth]{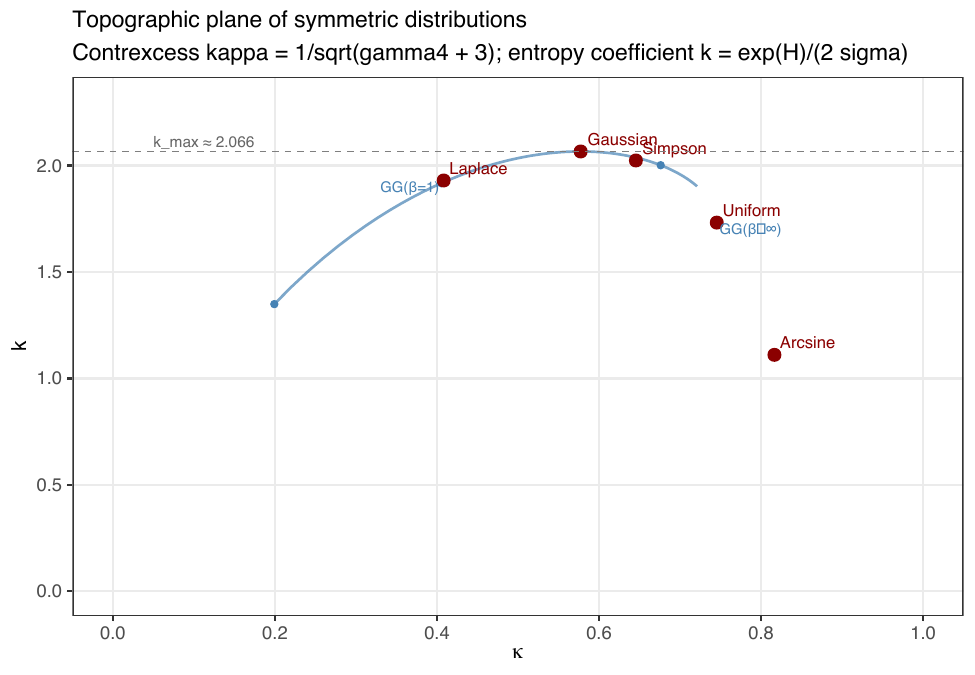}
  \caption{Topographic plane $(\varkappa, k)$ of symmetric distributions.
    Red dots~--- canonical distributions from the table; blue curve~---
    Generalized Gaussian family $\mathrm{GG}(\beta)$, passing through
    Laplace ($\beta = 1$), Gaussian ($\beta = 2$), and Uniform
    ($\beta \to \infty$). Dashed line~--- maximum possible value
    $k_{\max} = \sqrt{2\pi e}/2 \approx 2.066$ (Gaussian regime). Note:
    for the pair (Laplace, $\gamma_4 = 3$) and (Simpson, $\gamma_4 = -0.6$),
    projection onto kurtosis alone gives different points, while projection
    onto $k$~--- gives close values ($k \approx 1.93$ vs $2.02$), which
    motivates the use of $k$ as a second shape parameter.}
  \label{fig:topographic}
\end{figure}

\paragraph{Significance for PATP calibration.}
The topographic classification has methodological significance for PATP in
the finite-variance setting: two distributions with the same $\gamma_4$ but
different $k$ may be located in different \emph{optimality regions} of
PATP-$\alpha^{*}$. For practical applications (\S\ref{sec:algorithms}) this
means that calibrating $\alpha^{*}$ as a function of
$(\gamma_3, \gamma_4, k)$ may improve the estimate of the optimal $\alpha$
relative to the two-parameter $\alpha^{*}(\gamma_3, \gamma_4)$, especially in
the symmetric regime $\gamma_3 = 0$. For infinite-variance laws, $k$ must be
replaced by quantile, L-moment, or tail-index diagnostics. A numerical
investigation of the finite-variance advantage of $k$ is discussed in
\S\ref{sec:alg:calibration} and~\S\ref{sec:illustrations}.

\section{Parametrically Adaptive Transition Polynomial: Formal Definition}\label{sec:patp}

This section formalizes the PATP construction as a continuously
$\alpha$-parameterized family of basis functions within the generalized
Kunchenko apparatus (\S\ref{sec:bg:generalized}). \S\ref{sec:patp:family}
defines the basis family $\{\varphi_i(\,\cdot\,;\alpha)\}$;
\S\ref{sec:patp:param} derives the quadratic parameterization $p_i(\alpha)$
from three canonical constraints; \S\ref{sec:patp:limits} analyzes the three
limiting cases; \S\ref{sec:patp:body} establishes the conditions for linear
independence and non-degeneracy of the polynomial body $\Delta_S(\alpha)$;
\S\ref{sec:patp:criterion} formulates the PATP optimality criterion in the
form of the normal equations
$\mathbf{F}_S(\alpha)\,\vec h^{*}(\alpha) = \vec b(\alpha)$, which depend
on the control parameter~$\alpha$.

\subsection{Sign-preserving basis family}\label{sec:patp:family}

Fix a centered random variable $\xi$ (i.e., $\E[\xi] = 0$ at the true
parameter value). Introduce the control parameter $\alpha \in [0, 1]$ and
the exponent function $p_i(\alpha): [0,1] \to \mathbb{R}_{+}$ (its specific
form is fixed in \S\ref{sec:patp:param}). The basis functions of the
PATP family are defined as follows:
\begin{equation}\label{eq:patp-basis-def}
  \varphi_1(\xi) \;\equiv\; \xi,
  \qquad
  \varphi_i(\xi; \alpha)
  \;:=\; \mathrm{sign}(\xi) \cdot |\xi|^{\,p_i(\alpha)},
  \qquad i = 2, 3, \ldots, S.
\end{equation}
The construction \eqref{eq:patp-basis-def} is a special case of the
sign-preserving fractional-power functions $g_p^{-}$ from~\eqref{eq:gp-def}.
The first function~$\varphi_1$ is fixed as the identity (the same linear
step as in the classical power basis); the continuous parameter~$\alpha$
affects only the exponents for $i \geq 2$. Figure~\ref{fig:phi2_basis}
visualizes $\varphi_2(\xi; \alpha)$ for $\alpha \in \{0, 0.25, 0.5,
0.75, 1\}$~--- one can see the odd symmetry of the basis, the fractal
smoothing at $\alpha = 0$, and the collapse to a linear function at
$\alpha = 1/2$.

\begin{figure}[!htb]
  \centering
  \includegraphics[width=0.85\linewidth]{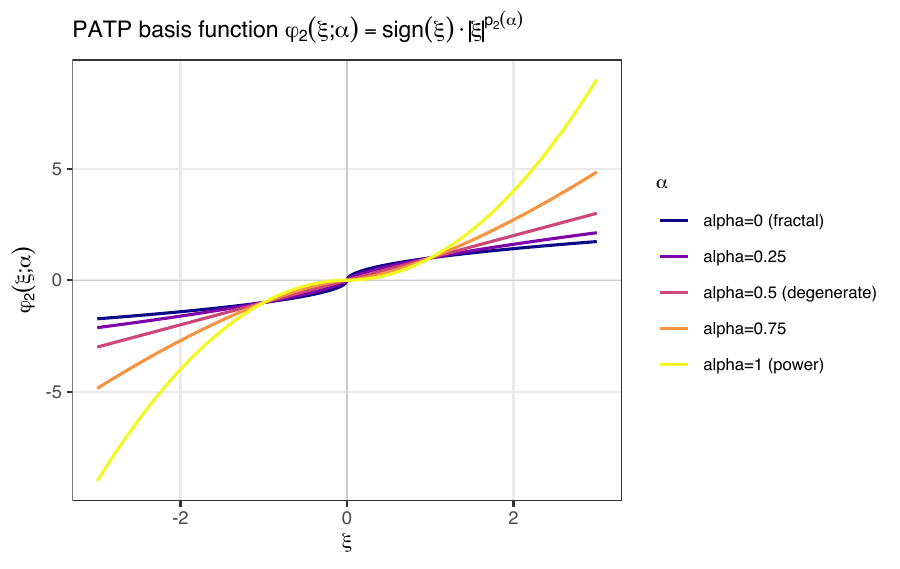}
  \caption{PATP basis function $\varphi_2(\xi; \alpha) = \mathrm{sign}(\xi)
    \cdot |\xi|^{p_2(\alpha)}$ for five values of $\alpha$. At $\alpha = 0$
    (fractal, $p_2 = 1/2$)~--- sublinear growth; at $\alpha = 1/2$
    (degenerate)~--- identity $\varphi_2 = \xi$; at $\alpha = 1$
    (signed-poly, $p_2 = 2$)~--- superlinear growth with \emph{odd}
    symmetry (in contrast to the canonical $\xi^2$, which is even).}
  \label{fig:phi2_basis}
\end{figure}

\paragraph{Why signed-parity (Form~B) and not Form~A.}
An earlier exploratory version of PATP applied the basis to the
\emph{observation} $y = R(\theta, x) + \xi$ as a whole (the so-called
Form~A, y-basis). This form is convenient for nonlinear regression since it
avoids explicit extraction of the residual, but has two structural
drawbacks. First, the matrix of centered correlants $\mathbf{F}_{ij}^{(v)}$
depends on the observation~$v$ through $R(\theta, x_v)$, which complicates
theoretical analysis. Second, for integer exponents Form~A theoretically
coincides with the power class of Kunchenko (\emph{universality property}),
but for fractional exponents this equivalence disappears and the gain from
the fractional extension is nullified.

\emph{Form~B} (residual-basis), corresponding to
\citep[ch.~4]{kunchenko2002}, works with
$\xi$~directly. The matrix $\mathbf{F}_S(\alpha)$ is \emph{scalar} (does
not depend on~$v$), and the derivation of $g_S(\alpha)$ in~\S\ref{sec:efficiency}
leads to a closed-form formula via the fractional absolute moments~$\nu_p$
from~\eqref{eq:fracmoment-def}. Therefore Form~B is chosen as the foundation
of the PATP apparatus in the present paper; Form~A remains suitable for
specific problems where $R(\theta, x)$ is linear and explicit (in
particular, for Royston-Altman fractional-polynomial regression settings
~\citep{royston1994}).

\subsection{Polynomial parameterization of exponents}\label{sec:patp:param}

\paragraph{Canonical constraints.}
Constructively we require that the parameterization $p_i(\alpha)$
reproduce three structural points simultaneously for all $i \geq 2$:
\begin{equation}\label{eq:patp-constraints}
  p_i(0) = \frac{1}{i}, \qquad
  p_i\!\left(\tfrac{1}{2}\right) = 1, \qquad
  p_i(1) = i.
\end{equation}
The limiting cases $\alpha = 0$ and $\alpha = 1$ correspond to the
\emph{fractal} ($p_i = 1/i$) and signed-parity \emph{integer-power}
($p_i = i$) regimes respectively. The midpoint
$\alpha = \tfrac{1}{2}$ defines the \emph{degenerate linear} regime
$p_i = 1$ for all $i$, in which the basis functions
\eqref{eq:patp-basis-def} collapse to a single linear function.
This is not an arbitrary choice: the point $\alpha = \tfrac{1}{2}$ serves
as the \emph{threshold} between the fractal and signed-parity integer-power
regimes and plays a structural role in the degeneracy analysis of
$g_2(\alpha)$
(\S\ref{sec:eff:unimodal}).

\paragraph{Lagrange interpolation.}
The three conditions~\eqref{eq:patp-constraints} uniquely determine a
quadratic polynomial in~$\alpha$. Applying the standard Lagrange
interpolation formula at the nodes $\{0, \tfrac{1}{2}, 1\}$ with values
$\{1/i, 1, i\}$:
\begin{align}
  p_i(\alpha)
  &= \frac{1}{i}\cdot
     \frac{(\alpha - \tfrac{1}{2})(\alpha - 1)}{(0 - \tfrac{1}{2})(0 - 1)}
   + 1 \cdot
     \frac{(\alpha - 0)(\alpha - 1)}{(\tfrac{1}{2} - 0)(\tfrac{1}{2} - 1)}
   + i \cdot
     \frac{(\alpha - 0)(\alpha - \tfrac{1}{2})}{(1 - 0)(1 - \tfrac{1}{2})}
   \notag\\[1ex]
  &= \frac{2}{i}\,(\alpha - \tfrac{1}{2})(\alpha - 1)
     \;-\; 4\,\alpha(\alpha - 1)
     \;+\; 2i\,\alpha(\alpha - \tfrac{1}{2}).
   \label{eq:p-i-lagrange}
\end{align}
Expanding the brackets and grouping by powers of~$\alpha$ yields the
canonical form of the PATP parameterization:
\begin{equation}\label{eq:p-i-final}
  \boxed{\,
  p_i(\alpha)
  \;=\; \frac{1}{i}
        \;+\; \left(4 - i - \frac{3}{i}\right)\alpha
        \;+\; \left(2i - 4 + \frac{2}{i}\right)\alpha^{2}.
  \,}
\end{equation}
Direct substitution of $\alpha \in \{0, \tfrac{1}{2}, 1\}$ verifies the
conditions~\eqref{eq:patp-constraints}. The simpler linear bridge
$p_i(\alpha) = (1 - \alpha)/i + \alpha\cdot i$ agrees with
\eqref{eq:p-i-final} at $\alpha = 0$ and $\alpha = 1$, but the quadratic
form adds the degenerate point $\alpha = \tfrac{1}{2}$, which is needed for
the structural degeneracy analysis.

\paragraph{Connection with fractional-polynomial regression.}
The $\alpha = 0$ realization of~\eqref{eq:p-i-final} is naturally related to
Royston-Altman fractional-polynomial regression~\citep{royston1994}. In the
standard regression setting, exponents are chosen not continuously
via~$\alpha$, but from the finite set
$\mathcal{P} = \{-2, -1, -0.5, 0, 0.5, 1, 2, 3\}$, which corresponds to
a selective ``snapshot'' from the fractal regime~$\alpha = 0$.

\subsection{Three limiting cases}\label{sec:patp:limits}

We fix the constraints of the parameterization~\eqref{eq:p-i-final} at the
three structural points and briefly discuss their semantics. In
Figure~\ref{fig:p_i_alpha} the curves $p_i(\alpha)$ are shown for
$i \in \{2,3,4,5\}$, illustrating the quadratic continuous interpolation
between the fractal ($\alpha = 0$), degenerate ($\alpha = 1/2$), and
power-polynomial ($\alpha = 1$) regimes.

\begin{figure}[!htb]
  \centering
  \includegraphics[width=0.85\linewidth]{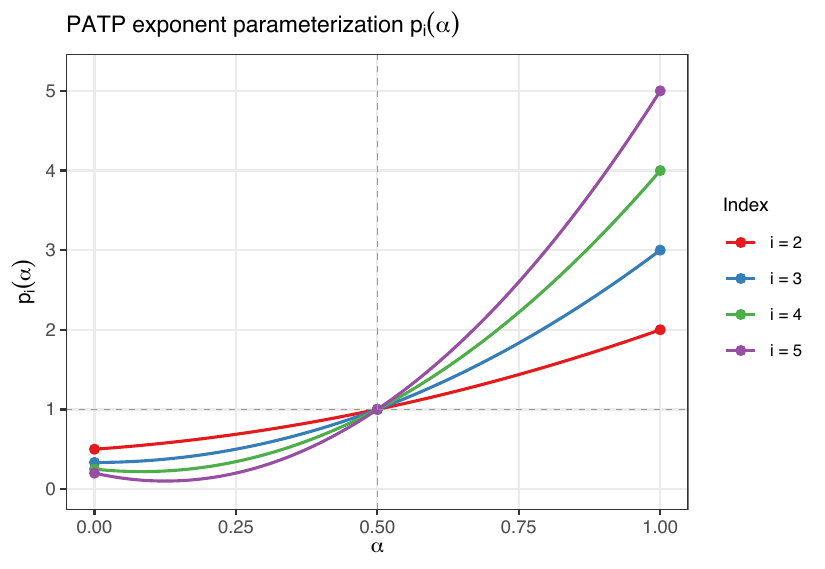}
  \caption{Quadratic curves $p_i(\alpha)$ for basis indices
    $i = 2, 3, 4, 5$. All curves pass through three structural points:
    $(0, 1/i)$ (fractal regime), $(1/2, 1)$ (degenerate point), and
    $(1, i)$ (classical Kunchenko power-polynomial regime). Dashed lines
    highlight the degenerate point $\alpha = 1/2$, where all $p_i$
    coincide with~1.}
  \label{fig:p_i_alpha}
\end{figure}

\paragraph{Case $\alpha = 0$: fractal regime.}
$p_i(0) = 1/i$, so the non-linear basis elements have the form
\[
  \varphi_i(\xi;0)=\mathrm{sign}(\xi)\,|\xi|^{1/i},
  \qquad i=2,\ldots,S.
\]
All basis functions grow \emph{slower} than linear, which makes the fractal
regime suitable for distributions with heavy tails: the low-order fractional
moments $\nu_{1/i}$ may exist even when the integer moments $\alpha_i$
diverge. This observation is not, by itself, sufficient for the $S=2$
variance formula, which still contains the baseline second moment $c_2$ and
higher fractional moments such as $\nu_{3/2}$; it therefore excludes the
Cauchy law. In this regime PATP is closely related to fractional lower-order
moment methods for heavy-tailed data~\citep{shao1993signal,matsui2013fracmoments}.

\paragraph{Case $\alpha = 1$: signed-parity integer-power regime.}
$p_i(1) = i$, so the non-linear basis elements have the form
\[
  \varphi_i(\xi;1)=\mathrm{sign}(\xi)\,|\xi|^i,
  \qquad i=2,\ldots,S.
\]
Through the sign-preserving construction~$\mathrm{sign}(\xi)\,|\xi|^i$ we
have an exact identity with the canonical powers $\xi^i$ for odd~$i$;
for even~$i$ the sign-preserving form gives $\mathrm{sign}(\xi)\,|\xi|^i$
instead of $\xi^i$. Although this difference appears formal at first glance,
its consequences for the matrix $\mathbf{F}_S$ are non-trivial (detailed
in \S\ref{sec:patp:body}). Consequently, the \emph{Form~B}-consistent
matrix~$\mathbf{F}_S$ at $\alpha=1$ should be treated as a signed-parity
integer-power analogue of the canonical PMM matrix, not as a literal
reproduction of the classical even-power construction.

\paragraph{Case $\alpha = \tfrac{1}{2}$: degenerate linear regime.}
$p_i(\tfrac{1}{2}) = 1$ for all $i \geq 2$, all basis functions collapse
to the linear one: $\varphi_i(\xi; \tfrac{1}{2}) = \xi$. The matrix
$\mathbf{F}_S(\tfrac{1}{2})$ is degenerate (rank 1), and the PMM estimator
reduces to OLS ($g_2(\tfrac{1}{2}) = 1$, no gain). This point serves as
a singular reference point for the function~$g_2(\alpha)$ and for the
numerical shape analysis in~\S\ref{sec:eff:unimodal}.

\subsection{Linear independence of the basis and the polynomial body}\label{sec:patp:body}

\paragraph{Non-degeneracy condition.}
For the PMM estimation problem to be well-posed (equation~\eqref{eq:normal-eq}),
the matrix $\mathbf{F}_S(\alpha)$ must be non-degenerate, i.e.,
$\Delta_S(\alpha) > 0$. This is equivalent to the \emph{linear independence}
of the basis functions $\{\varphi_i(\,\cdot\,;\alpha)\}_{i=1}^{S}$ on the
support of the distribution of~$\xi$.

\begin{proposition}\label{prop:linindep}
Let $\xi$~be a random variable with an absolutely continuous distribution
whose support includes a neighborhood of zero. Then the basis functions
$\{\varphi_i(\,\cdot\,;\alpha)\}_{i=1}^{S}$ from~\eqref{eq:patp-basis-def}
with $p_i(\alpha)$ given by \eqref{eq:p-i-final} are linearly independent
\emph{almost everywhere} if and only if all exponents $\{p_1, p_2(\alpha),
\ldots, p_S(\alpha)\}$ are distinct.
\end{proposition}

\begin{proof}[Proof (sketch)]
The converse is trivial: if $p_i(\alpha) = p_j(\alpha)$ for some
$i \neq j$, the corresponding basis functions coincide identically and
cannot be independent.

Direct statement. Suppose $\sum_{i=1}^{S} \lambda_i\,
\mathrm{sign}(\xi)|\xi|^{p_i} \equiv 0$ almost everywhere. Restricting to
the subset $\{\xi > 0\}$, we have $\sum \lambda_i |\xi|^{p_i} \equiv 0$ on
$(0, +\infty)$. By the classical result on the linear independence of
power functions (via the Vandermonde determinant in the corresponding
logarithmic coordinate), this is possible only if $\lambda_i = 0$ for
all~$i$, provided the exponents $p_i$ are distinct.
\end{proof}

\paragraph{Degeneracy at $\alpha = \tfrac{1}{2}$.}
By \eqref{eq:p-i-final} we have $p_i(\tfrac{1}{2}) = 1$ for all
$i \geq 2$, i.e., all exponents coincide with $p_1 \equiv 1$.
Proposition~\ref{prop:linindep} then fails, and
$\Delta_S(\tfrac{1}{2}) = 0$. This is the analytic manifestation of the
\emph{degenerate linear regime} (\S\ref{sec:patp:limits}): the PATP
estimator at $\alpha = \tfrac{1}{2}$ reduces to OLS.

\begin{theorem}[Exponent separation]\label{thm:exponent-separation}
For any distinct nonzero indices $i \neq j$, the equation
$p_i(\alpha)=p_j(\alpha)$ has the roots
\[
  \alpha=\frac{1}{2},
  \qquad
  \alpha=-\frac{1}{ij-1}.
\]
Consequently, on the admissible interval $[0,1]$ the only collision point of
distinct PATP exponents is $\alpha=\tfrac{1}{2}$.
\end{theorem}

\begin{proof}[Proof sketch]
Subtracting the two quadratic bridges and collecting terms gives
\[
  p_i(\alpha)-p_j(\alpha)
  =
  \frac{j-i}{ij}
  \left[
    1+(ij-3)\alpha-(2ij-2)\alpha^2
  \right].
\]
The bracketed quadratic factors as
\[
  -(2ij-2)
  \left(\alpha-\frac{1}{2}\right)
  \left(\alpha+\frac{1}{ij-1}\right).
\]
Since $i,j\geq 1$ and $i\neq j$, one has $ij>1$, so the second root is
negative and therefore lies outside $[0,1]$.
\end{proof}

\paragraph{No additional exponent-collision points on $[0,1]$.}
Theorem~\ref{thm:exponent-separation} replaces the earlier informal
``isolated degeneracies'' statement: the current PATP parametrization has no
extra interior exponent-collision points inside the admissible interval.

\paragraph{Polynomial body.}
Analogously to the classical definition~\citep{kunchenko2006stochastic} we
define the \emph{body of the PATP polynomial} as
\begin{equation}\label{eq:body-patp}
  \Delta_S(\theta; \alpha)
  \;:=\; \det\,\mathbf{F}_S(\theta; \alpha).
\end{equation}
The determinant is positive only under explicit moment and support
conditions. The following statement records the regular case used in the
rest of the paper.

\begin{theorem}[Positive definiteness of $\mathbf{F}_S(\alpha)$]\label{thm:positive-definite}
Let the distribution of $\xi$ have non-degenerate support containing an
interval, and let
\[
  \E\!\left[|\xi|^{p_i(\alpha)+p_j(\alpha)}\right] < \infty,
  \qquad i,j=1,\ldots,S.
\]
If $\alpha\in[0,1]\setminus\{\tfrac{1}{2}\}$, then the Form-B PATP basis
functions are linearly independent in $L_2(P)$ and the centered correlant
matrix $\mathbf{F}_S(\alpha)$ is positive definite.
\end{theorem}

\begin{proof}[Proof sketch]
The moment condition puts every basis product in $L_1(P)$ and every basis
function in $L_2(P)$. By Proposition~\ref{prop:linindep} and
Theorem~\ref{thm:exponent-separation}, the basis functions are distinct and
linearly independent on any support interval. For any nonzero coefficient
vector $\lambda$, the variance of
$\sum_i\lambda_i\varphi_i(\xi;\alpha)$ is therefore strictly positive, which
is exactly positive definiteness of the centered correlant matrix.
\end{proof}

At the boundary points $\alpha = 0$ and $\alpha = 1$ the elements of the
matrix $\mathbf{F}_S$ are expressed respectively via fractional and
integer-power moments, but the same finiteness and support qualifications
remain necessary. A compact list of Lean-verified algebraic facts supporting
the construction is given in Appendix~\ref{app:technical}.

\subsection{PATP optimality criterion}\label{sec:patp:criterion}

\paragraph{PATP stochastic polynomial.}
Substituting the basis family~\eqref{eq:patp-basis-def} into the general
construction~\eqref{eq:eta-S}, we obtain the PATP polynomial
\begin{equation}\label{eq:patp-poly}
  \eta_S(\xi; \theta, \alpha)
  \;=\; h_0
       \;+\; h_1\,\xi
       \;+\; \sum_{i=2}^{S} h_i\,
              \mathrm{sign}(\xi)\,|\xi|^{p_i(\alpha)},
\end{equation}
where $\theta \in \Theta$~is the unknown parameter, $\alpha \in [0,1]$~is
the basis control parameter, and $h_0, h_1, \ldots, h_S$~are non-random
coefficients belonging to the definition.

\paragraph{Maximum property and normal equations.}
The global maximum property~\eqref{eq:mmpl-property} is preserved for any
fixed $\alpha$:
\begin{equation}\label{eq:patp-max-property}
  \theta_0
  \;=\; \arg\max_{\theta \in \Theta}\;
        \E\bigl[\eta_S(\xi; \theta, \alpha)\bigr],
  \qquad \forall\,\alpha \in [0, 1]\setminus\{\tfrac{1}{2}\}.
\end{equation}
The corresponding normal equations for the optimal coefficients
$\vec h^{*}(\theta, \alpha)$:
\begin{equation}\label{eq:patp-normal-eq}
  \mathbf{F}_S(\theta; \alpha) \cdot \vec h^{*}(\theta, \alpha)
  \;=\; \vec b(\theta; \alpha),
\end{equation}
where $\mathbf{F}_S(\theta; \alpha)$~is the matrix of centered correlants
of the PATP basis, and the vector $\vec b(\theta; \alpha)$ has components
\begin{equation}\label{eq:b-patp}
  b_i(\theta; \alpha)
  \;=\; \frac{\partial \Psi_i(\theta; \alpha)}{\partial \theta},
  \qquad i = 1, \ldots, S.
\end{equation}

\paragraph{PATP estimator.}
The PATP estimator of the parameter $\theta$ for fixed $\alpha$ is defined as
\begin{equation}\label{eq:patp-estimator}
  \hat\theta_{\PATP,S}(\alpha)
  \;=\; \arg\max_{\theta \in \Theta}\;
        L_S\bigl(\theta; x_1, \ldots, x_N; \alpha\bigr),
\end{equation}
where $L_S(\theta; \cdots; \alpha)$~is the empirical functional
\eqref{eq:LS-empirical} with the PATP basis~\eqref{eq:patp-basis-def} and
coefficients $h_i^{*}(\theta; \alpha)$ from the system~\eqref{eq:patp-normal-eq}.

The canonical formula for the asymptotic variance of the
estimator~\eqref{eq:var-pmm-canonical} carries over to the PATP format:
\begin{equation}\label{eq:patp-variance}
  \Var\bigl[\hat\theta_{\PATP,S}(\alpha)\bigr]
  \;=\; \frac{1}{N \cdot
                \vec b^{\top}(\theta_0; \alpha)\,
                \mathbf{F}_S^{-1}(\theta_0; \alpha)\,
                \vec b(\theta_0; \alpha)},
\end{equation}
and the corresponding variance reduction coefficient $g_S(\alpha)$ is the
main object of theoretical analysis in~\S\ref{sec:efficiency}.

\paragraph{Adaptive choice of $\alpha$.}
The PATP construction stands out among neighboring semiparametric frameworks
precisely through the possibility of \emph{adapting} the control parameter
$\alpha$ to the estimated characteristics of the noise distribution.
The optimal $\alpha^{*}$ is expected to be a function of the standardized
cumulants:
\begin{equation}\label{eq:alpha-optimal}
  \alpha^{*}
  \;=\; \arg\min_{\alpha \in [0, 1]}\; g_S(\alpha; \gamma_3, \gamma_4),
\end{equation}
which yields the algorithmic scheme ``OLS $\to$ estimate $\gamma_3, \gamma_4$
from residuals $\to$ calibrate $\alpha^{*}$ $\to$ PATP estimate with chosen
$\alpha^{*}$'' (detailed in~\S\ref{sec:alg:calibration}).
The unimodality of the function $g_2(\alpha)$ and the corresponding
theoretical results on the optimal $\alpha^{*}$ are developed
in~\S\ref{sec:eff:unimodal}.

\section{Theoretical efficiency analysis}\label{sec:efficiency}

This section derives a closed-form formula for the variance reduction
coefficient $g_2(\alpha)$ of the PATP estimator of degree $S = 2$ in the
problem of estimating a location parameter from residuals~$\xi$.
\S\ref{sec:eff:gs} formalizes the general definition of $g_S(\alpha)$ via
the variance~\eqref{eq:patp-variance}. \S\ref{sec:eff:s2} develops the
computation of the matrix $\mathbf{F}_2(\alpha)$ and the vector
$\vec b(\alpha)$ through the fractional absolute moments~$\nu_q$ and
derives the closed-form formula~$g_2(\alpha)$. \S\ref{sec:eff:unimodal}
analyzes the behavior of $g_2$ as a function of~$\alpha$ and derives the
degeneracy at $\alpha = 1/2$. \S\ref{sec:eff:heavy} discusses the
advantage of the fractal regime $\alpha = 0$ in the heavy-tail setting
through the exclusively fractional nature of the dependence~$g_2$.

\subsection[Variance reduction coefficient gS(alpha)]{Variance reduction coefficient $g_S(\alpha)$}\label{sec:eff:gs}

From~\eqref{eq:patp-variance} the asymptotic variance of the PATP estimator
of degree~$S$ for fixed $\alpha \in [0,1] \setminus \{1/2\}$:
\begin{equation}\label{eq:var-patp-2}
  \Var\bigl[\hat\theta_{\PATP,S}(\alpha)\bigr]
  \;=\; \frac{1}{N \cdot
                \vec b^{\top}(\alpha)\,
                \mathbf{F}_S^{-1}(\alpha)\,
                \vec b(\alpha)}.
\end{equation}
The variance of the OLS estimator of the mean is $\Var[\hat\mu_{\mathrm{OLS}}] = c_2/N$.
Hence
\begin{equation}\label{eq:gS-alpha-def}
  g_S(\alpha)
  \;:=\; \frac{\Var\bigl[\hat\theta_{\PATP,S}(\alpha)\bigr]}
              {\Var\bigl[\hat\mu_{\mathrm{OLS}}\bigr]}
  \;=\; \frac{1}{c_2 \cdot \vec b^{\top}(\alpha)\,
                \mathbf{F}_S^{-1}(\alpha)\,\vec b(\alpha)}.
\end{equation}
The smaller $g_S(\alpha)$ is for a given distribution of $\xi$, the more
efficient the PATP estimator is relative to OLS. In particular,
$g_S(\alpha) = 1$ indicates no gain (degenerate regime or Gaussian noise).

\subsection[Specialization to S = 2: closed-form formula]{Specialization to $S = 2$: closed-form formula}\label{sec:eff:s2}

\paragraph{Location and centering convention.}
The formulas below use residuals
\[
  \xi = X-\theta_0,
\]
where $\theta_0$ is the target location parameter. In a finite-mean setting
one may take $\theta_0=\E[X]$ and hence $\E[\xi]=0$. In an infinite-mean
location-family setting, however, $\theta_0$ must be supplied by the model
location, median, or another robust center; statements involving OLS variance
or classical mean estimation are then not meaningful without additional
qualification. This is why Cauchy is treated in this paper as a limiting
counterexample for the $S=2$ variance formula rather than as a covered
finite-mean case.

\paragraph{Notation.}
For convenience we introduce
\begin{equation}\label{eq:p-notation}
  p \;:=\; p_2(\alpha)
  \;=\; \tfrac{1}{2} + \tfrac{1}{2}\alpha + \alpha^2,
\end{equation}
i.e., the second exponent of the PATP parameterization~\eqref{eq:p-i-final}
at $i = 2$. Direct substitution verifies: $p(0) = 1/2$, $p(1/2) = 1$,
$p(1) = 2$.

\paragraph{Fractional moments.}
For residuals centered in the above sense we introduce the absolute and
signed fractional moments:
\begin{equation}\label{eq:moments-q}
  \nu_q \;:=\; \E\bigl[|\xi|^q\bigr],
  \qquad
  \sigma_q \;:=\; \E\bigl[\mathrm{sign}(\xi)\,|\xi|^q\bigr].
\end{equation}
The required range of $q$ is not arbitrary: in the $S=2$ formula below one
needs $\nu_{p-1}$, $\nu_{p+1}$, $\nu_{2p}$, and $\sigma_p$ to exist. When
$\E[\xi]=0$, $\nu_2 = c_2$ (variance), $\sigma_3 = c_3$ (third central
moment), and $\nu_1$ is the mean absolute deviation. For symmetric
distributions $\sigma_q = 0$ for all $q$ for which the expectation exists.

\begin{theorem}[Moment and regularity conditions for $g_2(\alpha)$]\label{thm:g2-conditions}
Fix $\alpha\in[0,1]\setminus\{\tfrac{1}{2}\}$ and put
$p=p_2(\alpha)$. Formula~\eqref{eq:g2-alpha} is meaningful if:
\begin{enumerate}
  \item $c_2=\E[\xi^2]$, $\nu_{p-1}$, $\nu_{p+1}$, $\nu_{2p}$, and
        $\sigma_p$ are finite;
  \item when $p<1$, the law has no atom at zero and is locally regular
        enough near zero for $\E[|\xi|^{p-1}]$ to be finite
        (a bounded density near zero is sufficient; see
        Appendix~\ref{app:local-integrability});
  \item the determinant
        $\Delta_2(\alpha)=c_2(\nu_{2p}-\sigma_p^2)-\nu_{p+1}^2$ is positive.
\end{enumerate}
Because the current $S=2$ construction keeps the linear basis
$\varphi_1(\xi)=\xi$, it still requires finite variance.  The fractal branch
therefore lowers the nonlinear moment order relative to the classical PMM2
fourth-moment requirement, but it does not cover infinite-variance laws.
\end{theorem}

\begin{proof}[Proof sketch]
The matrix entries in $\mathbf{F}_2(\alpha)$ require the second-order
products $\xi\varphi_2(\xi;\alpha)$ and $\varphi_2^2(\xi;\alpha)$ to be
integrable, which gives $\nu_{p+1}$ and $\nu_{2p}$ plus the signed moment
$\sigma_p$. Differentiating the signed-power basis with respect to location
introduces $|\xi|^{p-1}$, hence $\nu_{p-1}$ and the local condition near
zero when $p<1$.  The term $F_{11}=c_2$ is the reason finite variance is
still required for the present $S=2$ formula. Positive determinant is the
$S=2$ instance of the positive-definiteness requirement from
Theorem~\ref{thm:positive-definite}.
\end{proof}

\paragraph{PATP basis of degree 2.}
By~\eqref{eq:patp-basis-def}: $\varphi_1(\xi) = \xi$,
$\varphi_2(\xi; \alpha) = \mathrm{sign}(\xi) |\xi|^{p}$. Expected values:
\begin{equation}\label{eq:psi-1-2}
  \Psi_1 \;=\; \E[\xi] = 0,
  \qquad
  \Psi_2(\alpha) \;=\; \E[\mathrm{sign}(\xi)|\xi|^p] \;=\; \sigma_p.
\end{equation}

\paragraph{Matrix $\mathbf{F}_2(\alpha)$.}
Computing sequentially via~\eqref{eq:F-matrix}:
\begin{align}
  F_{11}(\alpha)
    &= \E[\xi^2] - 0
     = c_2,
  \label{eq:F11-patp}\\
  F_{12}(\alpha) = F_{21}(\alpha)
    &= \E\bigl[\xi \cdot \mathrm{sign}(\xi)|\xi|^p\bigr] - 0 \cdot \sigma_p
     = \E\bigl[|\xi| \cdot |\xi|^p\bigr]
     = \nu_{p+1},
  \label{eq:F12-patp}\\
  F_{22}(\alpha)
    &= \E\bigl[(\mathrm{sign}(\xi)|\xi|^p)^2\bigr] - \sigma_p^2
     = \E\bigl[|\xi|^{2p}\bigr] - \sigma_p^2
     = \nu_{2p} - \sigma_p^2.
  \label{eq:F22-patp}
\end{align}
Note that $F_{12}$ is expressed through the \emph{symmetric} fractional
moment $\nu_{p+1}$ (rather than the signed $\sigma_{p+1}$): the product
$\xi \cdot \mathrm{sign}(\xi) = |\xi|$ reduces the parity of the product
to~$|\xi|^{p+1}$, which is a structural feature of Form~B with
signed-parity. The determinant of the matrix:
\begin{equation}\label{eq:detF-patp}
  \Delta_2(\alpha)
  \;:=\; \det \mathbf{F}_2(\alpha)
  \;=\; c_2(\nu_{2p} - \sigma_p^2) - \nu_{p+1}^2.
\end{equation}

\paragraph{Vector $\vec b(\alpha)$.}
The components of the vector~$\vec b$ are computed via~\eqref{eq:b-patp}.
Taking the derivative with respect to~$\mu$ ($\xi = y - \mu$) and using
$\partial_{\xi}[\mathrm{sign}(\xi)|\xi|^p] = p|\xi|^{p-1}$
(an odd basis function with an even derivative almost everywhere),
we obtain:
\begin{align}
  b_1 \;=\; \partial_{\mu} \Psi_1
     &= -1,\\
  b_2(\alpha) \;=\; \partial_{\mu} \Psi_2(\alpha)
     &= -p \cdot \E\bigl[|\xi|^{p-1}\bigr]
     \;=\; -p \nu_{p-1}.
\end{align}
The sign~$-1$ can be absorbed into a redefinition of the coefficients~$h_i$
without affecting the quadratic form~\eqref{eq:gS-alpha-def}; henceforth
we work with
\begin{equation}\label{eq:b-patp-final}
  \vec b(\alpha)
  \;=\;
  \begin{pmatrix} 1 \\ p\,\nu_{p-1} \end{pmatrix}.
\end{equation}

\paragraph{Closed-form formula $g_2(\alpha)$.}
Substituting the matrix~\eqref{eq:F11-patp}--\eqref{eq:F22-patp} and the
vector~\eqref{eq:b-patp-final} into~\eqref{eq:gS-alpha-def}, we first
compute the quadratic form:
\begin{align}
  \vec b^{\top}\,\mathbf{F}_2^{-1}\,\vec b
   &= \frac{1}{\Delta_2}\bigl[\,F_{22} - 2 F_{12}\, p\,\nu_{p-1}
                              + F_{11}\,(p\,\nu_{p-1})^2\bigr]
   \notag\\
   &= \frac{(\nu_{2p} - \sigma_p^2) \,-\, 2 p\,\nu_{p+1}\,\nu_{p-1}
              \,+\, p^2\,c_2\,\nu_{p-1}^2}
           {c_2(\nu_{2p} - \sigma_p^2) - \nu_{p+1}^2}.
  \label{eq:bFb-patp}
\end{align}
Hence
\begin{equation}\label{eq:g2-alpha}
  \boxed{\,
  g_2(\alpha)
  \;=\;
  \frac{c_2(\nu_{2p} - \sigma_p^2) - \nu_{p+1}^2}
       {c_2\bigl[\,(\nu_{2p} - \sigma_p^2)
                  - 2 p\,\nu_{p+1}\,\nu_{p-1}
                  + p^2\,c_2\,\nu_{p-1}^2\,\bigr]},
  \qquad p = p_2(\alpha).
  \,}
\end{equation}
Formula~\eqref{eq:g2-alpha} is the central result of the present paper:
an explicit dependence of $g_2$ on $\alpha$ through the single exponent
$p = p_2(\alpha)$ and four moments $\nu_{p-1}, \nu_{p+1}, \nu_{2p},
\sigma_p$.

\paragraph{Verification in the signed-parity integer-power regime ($\alpha = 1$).}
At $\alpha = 1$ we have $p = 2$, so
$\nu_{p-1} = \nu_1$, $\nu_{p+1} = \nu_3$, $\nu_{2p} = \nu_4 = m_4 = c_4 + 3c_2^2$,
$\sigma_p = \sigma_2 = \E[\mathrm{sign}(\xi)\xi^2] = \E[\xi|\xi|]$
(the asymmetric second moment). Substituting into~\eqref{eq:g2-alpha},
we obtain the PATP analog $g_2(1)$, which does \emph{not} reduce literally
to the classical formula $1 - \gamma_3^2/(2+\gamma_4)$
from~\eqref{eq:g2-standardized}, since the PATP construction uses the
signed-parity function $\mathrm{sign}(\xi)\xi^2 = \xi|\xi|$ (odd) instead
of the canonical $\xi^2$ (even). For symmetric distributions
($\sigma_2 = 0$), formula~\eqref{eq:g2-alpha} at $\alpha=1$ becomes
\begin{equation}\label{eq:g2-symmetric}
  g_2(1; \text{symmetric})
  \;=\;
  \frac{c_2\nu_4-\nu_3^2}
       {c_2\left(\nu_4-4\nu_3\nu_1+4c_2\nu_1^2\right)} ,
\end{equation}
(distinct from the classical $g_2 = 1$ for symmetric distributions), which
shows that the signed-parity endpoint has different efficiency behavior even
when classical skewness is zero. Here $\nu_1,\nu_3,\nu_4$ are absolute
moments; only the signed moments vanish by symmetry.

\paragraph{No-gain condition: connection to Stein's identity.}
The PATP estimator yields no gain over OLS ($g_2(\alpha) = 1$) if and only if
the optimal coefficient on $\varphi_2$ vanishes, i.e.\ $F_{11}\,b_2 = F_{12}\,b_1$,
which by~\eqref{eq:F11-patp}-\eqref{eq:F12-patp} and~\eqref{eq:b-patp-final} is
equivalent to the identity
\begin{equation}\label{eq:stein-condition}
  \nu_{p+1} \;=\; p\,c_2\,\nu_{p-1}
  \qquad\Longleftrightarrow\qquad
  \E\bigl[\xi\,\varphi_2(\xi)\bigr] \;=\; c_2\,\E\bigl[\varphi_2'(\xi)\bigr].
\end{equation}
This is precisely \emph{Stein's identity} for the signed-parity function
$\varphi_2(\xi) = \mathrm{sign}(\xi)|\xi|^p$. For Gaussian noise
$\xi \sim \mathcal{N}(0, \sigma^2)$ it holds for \emph{every} $p$ (hence for
every $\alpha$), which immediately yields $g_2(\alpha) \equiv 1$: the Gaussian
case admits no gain in any PATP regime, consistent with the efficiency of the
sample mean. For any non-Gaussian symmetric noise (Laplace, generalized
Gaussian, EPD), identity~\eqref{eq:stein-condition} fails, so a gain persists
even at the integer-power endpoint $\alpha = 1$.

\paragraph{Verification in the fractal regime ($\alpha = 0$).}
At $\alpha = 0$ we have $p = 1/2$, so
$\nu_{p-1} = \nu_{-1/2}$, $\nu_{p+1} = \nu_{3/2}$,
$\nu_{2p} = \nu_1$, $\sigma_p = \sigma_{1/2}$. Apart from the baseline
second moment $c_2$ inherited from $\varphi_1(\xi)=\xi$, the nonlinear
branch now involves only fractional orders $-1/2$, $1$, and $3/2$.  The
near-zero term $\nu_{-1/2}$ is finite under ordinary bounded-density
regularity (Appendix~\ref{app:local-integrability}).  Thus the $S=2$
fractal endpoint relaxes the classical fourth-moment requirement to a
finite-variance requirement; infinite-variance stable laws and Cauchy remain
outside this particular variance formula.

\subsection[Degeneracy and behavior of g2(alpha)]{Degeneracy and behavior of $g_2(\alpha)$}\label{sec:eff:unimodal}

\paragraph{Degeneracy at $\alpha = 1/2$.}
At $\alpha = 1/2$: $p(1/2) = 1$, so
$\nu_{p+1} = \nu_2 = c_2$, $\nu_{2p} = \nu_2 = c_2$,
$\sigma_p = \sigma_1 = \E[\mathrm{sign}(\xi)|\xi|] = \E[\xi] = 0$
(centered), $\nu_{p-1} = \nu_0 = 1$. Substituting into the determinant:
\begin{equation}\label{eq:det-degenerate}
  \Delta_2(1/2) \;=\; c_2 \cdot c_2 \,-\, c_2^2 \;=\; 0,
\end{equation}
i.e., the matrix $\mathbf{F}_2(1/2)$ is degenerate. This is the analytic
manifestation of the structural observation in~\S\ref{sec:patp:body}: at
$\alpha = 1/2$ the PATP family collapses to a single linear function,
$\mathbf{F}_2$ has rank~1, and the PATP estimator reduces to OLS. Both the
numerator and denominator of formula~\eqref{eq:g2-alpha} vanish, giving an
indeterminate form $0/0$; the limit $\lim_{\alpha \to 1/2} g_2(\alpha) = 1$
is computed by L'Hôpital's rule.

\paragraph{Limiting behavior.}
Substituting the boundary values $\alpha = 0$ and $\alpha = 1$
into~\eqref{eq:g2-alpha} and using the cancellation of fractional moments,
we obtain explicit formulas $g_2(0)$ and $g_2(1)$ as functions of the
characteristics of the noise distribution. For the symmetric canonical laws of
this paper (Laplace, Generalized Gaussian, Student-$t$) these values are
computed in closed form and presented in~\S\ref{sec:illustrations}.

\paragraph{Numerical curve shape on $[0, 1] \setminus \{1/2\}$.}
Numerical sweeps of~\eqref{eq:g2-alpha} show a structured dependence of the
minimizer $\alpha^{*}$ on the shape of the noise distribution. The adaptive
selection of $\alpha^{*}$ does \emph{not} rely on a unimodality theorem: in
practice $\alpha^{*}$ is obtained by direct minimization of the closed-form
$g_2(\alpha)$ (or its plug-in estimate) over a grid on $[0,1]\setminus\{1/2\}$,
as specified in~\S\ref{sec:alg:calibration}. A general analytic proof of
unimodality remains an open problem, since formula~\eqref{eq:g2-alpha} is a
rational function of four integrals $\nu_{p-1}, \nu_{p+1}, \nu_{2p}, \sigma_p$,
each of which depends on $\alpha$ through $p_2(\alpha)$ nonlinearly; we
therefore present the curve shape as an empirical observation on the symmetric
canonical distributions of~\S\ref{sec:illustrations} (Laplace, Generalized
Gaussian, Student-$t$), and claim no general unimodality theorem.

\subsection{Heavy-tail behavior}\label{sec:eff:heavy}

\paragraph{Structural advantage of the fractal regime.}
The fundamental difference between PATP and the classical Kunchenko PMM
manifests in the \emph{integrability requirements}. The classical
formula~\eqref{eq:g2-standardized} requires finiteness of the fourth
moment $\E[|\xi|^4] = \nu_4 < \infty$. This is violated by many
finite-variance heavy-tailed distributions, for example Student or Pareto
families with finite second but infinite fourth moment, and is often
empirically unstable in financial or telecommunication data.

By contrast, the PATP formula~\eqref{eq:g2-alpha} at $\alpha = 0$ involves
the baseline second moment $c_2$ and fractional nonlinear orders
$\{-1/2, 1, 3/2\}$. This expands the admissible class from fourth-moment
models to finite-variance models:

\begin{itemize}
  \item \textbf{Laplace} ($\nu_q$ finite for all $q > -1$):
        PATP works both at $\alpha = 0$ and at $\alpha = 1$.
  \item \textbf{Finite-variance, infinite-fourth-moment tails}:
        PATP at $\alpha = 0$ can remain meaningful because it requires
        $c_2$ and fractional moments up to order $3/2$, but not $\nu_4$.
  \item \textbf{Infinite-variance laws} (Cauchy or stable laws with
        $\alpha_s<2$): the present $S=2$ variance formula is not meaningful
        because $c_2$ and the OLS reference variance are infinite.  These
        cases require a different location functional or a modified basis
        without the unsmoothed linear component.
\end{itemize}

\paragraph{Observed shape of the curve $g_2(\alpha)$.}
Based on~\eqref{eq:g2-alpha} and the singular point at $\alpha=1/2$, the
canonical examples show a recurring two-sided pattern: the limit near
$\alpha=1/2$ corresponds to the OLS reference level, while one or both sides
of the interval may contain lower values depending on the distribution.
Figure~\ref{fig:g2_alpha} illustrates this behavior for five canonical
distributions. For heavy-tailed distributions
(Laplace, GG(0.5)) the global minimum of $g_2$ is attained in the fractal
regime $\alpha \to 0$; for light-tailed distributions (Uniform, GG(4))~---
in the signed-parity integer-power regime $\alpha \to 1$. The Gaussian case gives
$g_2 \equiv 1$ (a flat horizontal line), corresponding to no gain over OLS.

\begin{figure}[!htb]
  \centering
  \includegraphics[width=0.9\linewidth]{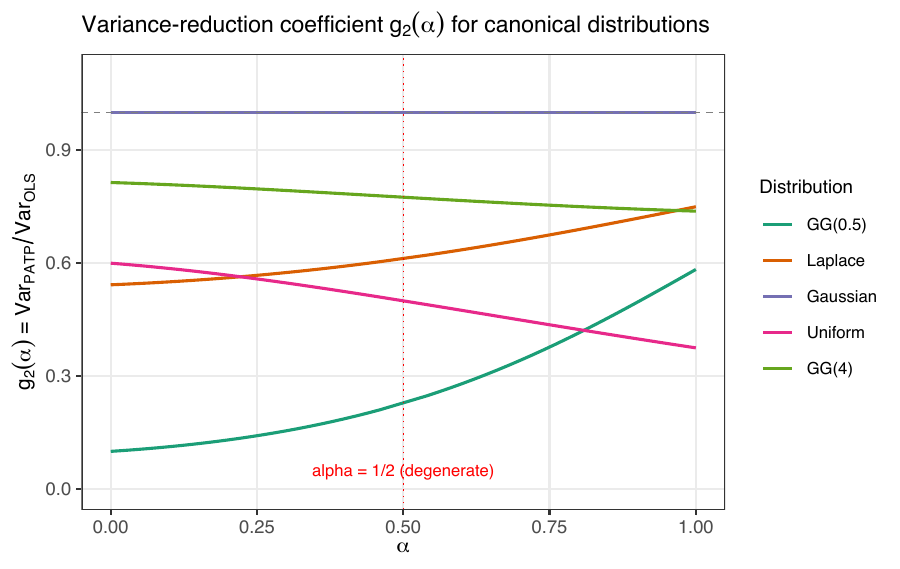}
  \caption{Theoretical variance reduction coefficient
    $g_2(\alpha) = \Var[\hat\mu_{\PATP}]/\Var[\hat\mu_{\mathrm{OLS}}]$
    for five canonical distributions, computed via the closed-form
    formula~\eqref{eq:g2-alpha}. The red dashed line marks the degenerate
    point $\alpha = 1/2$, where $g_2 \to 1$. A clear pattern is observed:
    for heavy tails (GG(0.5), Laplace) the optimum $\alpha^{*}$ shifts
    toward the fractal regime $\alpha = 0$; for light tails (Uniform,
    GG(4))~--- toward the power-polynomial regime $\alpha = 1$. The
    Gaussian case gives $g_2 \equiv 1$ (no gain).}
  \label{fig:g2_alpha}
\end{figure}

\paragraph{Formal verification of boundary cases.}
Three structural theorems about the boundary cases of the PATP exponents
($p_i(0) = 1/i$, $p_i(1/2) = 1$, $p_i(1) = i$) have been formalized and
machine-verified in Lean~4 (module \texttt{PATP.Param},
file \texttt{Lean/PATP/Param.lean}). Analogous theorems about the boundary
cases of the basis itself ($\varphi_i(\xi; 1/2) = \xi$ for all $i$;
\emph{odd symmetry} $\varphi_i(-\xi; \alpha) = -\varphi_i(\xi; \alpha)$
for all $i, \alpha$) have been formalized in module \texttt{PATP.Basis}.
Deriving formula~\eqref{eq:g2-alpha} in Lean (module
\texttt{PATP.G2Algebra}) is currently restricted to the symbolic algebra
from $\vec b^{\top}\mathbf F_2^{-1}\vec b$ to the displayed rational
expression. Formalizing the expectations as integrals with respect to a
probability measure remains beyond the current Lean layer; the exact scope
of the formalization is summarized in Appendix~\ref{app:lean-facts}.

\section{Implementation Algorithms}\label{sec:algorithms}

This section describes the computational implementation of the PATP estimator
$\hat\theta_{\PATP,S}(\alpha)$ from~\eqref{eq:patp-estimator}.
\S\ref{sec:alg:nr} describes the fixed-$\alpha$ score equation and the
solver stack used to avoid singular derivative behavior. \S\ref{sec:alg:calibration}
formulates calibration of $\alpha^{*}$ from fractional-moment criteria and
robust shape diagnostics. \S\ref{sec:alg:complexity} discusses computational
complexity and numerical aspects.

\subsection[Fixed-alpha score equation and solver stack]{Fixed-$\alpha$ score equation and solver stack}\label{sec:alg:nr}

\paragraph{Score function.}
For fixed $\alpha$ and optimal coefficients
$\vec h^{*}(\theta; \alpha)$ from~\eqref{eq:patp-normal-eq}, differentiating
the empirical functional~\eqref{eq:LS-empirical} with respect to $\theta$ gives
the \emph{PATP score function}
\begin{equation}\label{eq:patp-score}
  Z_N(\theta; \alpha)
  \;:=\; \frac{\partial L_S(\theta; x_1, \ldots, x_N; \alpha)}
              {\partial \theta}
  \;=\; \sum_{i=1}^{S} h_i^{*}(\theta; \alpha)
        \cdot \frac{1}{N} \sum_{n=1}^{N}
        \frac{\partial \varphi_i(\xi_n; \alpha)}{\partial \theta},
\end{equation}
where $\xi_n = x_n - R(\theta, z_n)$ (for the regression model) or
$\xi_n = x_n - \theta$ (for mean estimation). The estimate
$\hat\theta_{\PATP}(\alpha)$ is the solution of the equation
$Z_N(\theta; \alpha) = 0$.

\paragraph{Newton iterations.}
When the score is smooth and the derivative is numerically stable, one may
use the Newton-Raphson scheme for solving $Z_N(\theta; \alpha) = 0$:
\begin{equation}\label{eq:nr-iteration}
  \theta^{(k+1)}
  \;=\; \theta^{(k)} \;-\;
        \frac{Z_N\bigl(\theta^{(k)}; \alpha\bigr)}
             {Z_N'\bigl(\theta^{(k)}; \alpha\bigr)},
\end{equation}
with initial approximation $\theta^{(0)} = \hat\theta_{\mathrm{OLS}}$
(the OLS estimate) and convergence criterion $|\theta^{(k+1)} - \theta^{(k)}|
< \varepsilon$ for a pre-specified $\varepsilon > 0$
(typically $\varepsilon = 10^{-6} \cdot |\theta^{(0)}|$).

\paragraph{Derivatives of basis functions.}
For the PATP basis~\eqref{eq:patp-basis-def} the derivatives take the form
away from $\xi_n=0$:
\begin{equation}\label{eq:basis-derivatives}
  \frac{\partial \varphi_1}{\partial \theta} = -\frac{\partial R}{\partial \theta},
  \qquad
  \frac{\partial \varphi_i(\xi_n; \alpha)}{\partial \theta}
  \;=\; -p_i(\alpha) \cdot |\xi_n|^{p_i(\alpha) - 1}
        \cdot \frac{\partial R}{\partial \theta},
  \quad i \geq 2.
\end{equation}
Note that the derivative of the signed-parity basis function is \emph{even}
in $\xi_n$ (due to the doubling of the sign function): this is structurally
different from the classical $\partial(\xi^i)/\partial\theta = -i\,\xi^{i-1}
\partial R/\partial\theta$, which preserves parity ($i-1$ even $\Rightarrow$
odd; $i-1$ odd $\Rightarrow$ even).

\paragraph{Regularity and fallback solvers.}
For $p_i(\alpha)<1$, the derivative
$|\xi|^{p_i(\alpha)-1}$ is singular at zero. The Newton formula should
therefore be treated as the fast path, not as the only valid algorithm. The
implementation stack used in the experiments is:
\begin{enumerate}
  \item use bracketing root-finding for scalar location problems whenever the
        score changes sign on a robust interval around the preliminary
        center;
  \item use damped Newton only when the derivative is finite and the update
        decreases the absolute score;
  \item replace $|\xi|^p$ by a smoothed form
        $(\xi^2+\varepsilon^2)^{p/2}$ near zero when repeated zeros or
        discretized data make the derivative unstable;
  \item fall back to OLS or the chosen robust center when
        $|\alpha-\tfrac{1}{2}|<\delta$, because
        $\mathbf{F}_S(\alpha)$ is nearly singular.
\end{enumerate}
This regularity convention is the algorithmic counterpart of
Theorem~\ref{thm:g2-conditions}: differentiability is required only almost
everywhere, but numerical solvers must still protect the neighborhood of zero.

\paragraph{Full $\mathbf{F}_2^{-1}\mathbf{b}$ PATP solver.}
Algorithm~\ref{alg:patp-full} below specifies the full normal-equation
estimator used in the reproducible R pipeline (file
\texttt{R/05\_full\_patp\_estimator.R}). It is a one-step Newey
linearisation: the optimal weights $\mathbf{h}^{*}=\mathbf{F}_2^{-1}\mathbf{b}$
are computed empirically at an initial OLS centre and the location is then
updated by a single Newton step on the PATP score; an optional outer loop
re-builds $\mathbf{F}_2,\mathbf{b}$ at the updated centre and re-iterates
until the step shrinks below a target tolerance. The scalar signed-power
proxy of \eqref{eq:patp-score} is retained as an automatic fallback when
$\mathbf{F}_2$ is ill-conditioned.

\begin{algorithm}
\caption{PATP-2 full $\mathbf{F}_2^{-1}\mathbf{b}$ estimator.}\label{alg:patp-full}
\begin{algorithmic}[1]
\Require sample $x_{1:N}$; control parameter $\alpha\in[0,1]\setminus\{\tfrac{1}{2}\}$;
         tolerance $\varepsilon$; max iterations $K$; condition cap $\kappa_{\max}$
\Ensure location estimate $\hat\mu$
\State $p\gets p_2(\alpha)$; $\hat\mu\gets\bar x$
\For{$k=1,\ldots,K$}
  \State $\xi_n\gets x_n-\hat\mu$; compute $\hat\nu_{p-1},\hat\nu_{p+1},\hat\nu_{2p},
         \hat\sigma_p,\hat c_2$
  \State $\mathbf{F}_2\gets\begin{pmatrix}\hat c_2 & \hat\nu_{p+1}\\
                             \hat\nu_{p+1} & \hat\nu_{2p}-\hat\sigma_p^2\end{pmatrix}$,\quad
         $\mathbf{b}\gets(1,\,p\hat\nu_{p-1})^{\top}$
  \If{$\kappa(\mathbf{F}_2)>\kappa_{\max}$ \textbf{or} $|\det\mathbf{F}_2|<10^{-14}$}
    \State \Return \Call{ProxyScalar}{$x_{1:N},\alpha,\hat\mu$} \Comment{fallback}
  \EndIf
  \State $\mathbf{h}^{*}\gets\mathbf{F}_2^{-1}\mathbf{b}$
  \State $Z\gets h_1^{*}\bar\xi+h_2^{*}\hat\sigma_p$;\quad
         $Z'\gets -h_1^{*}-p h_2^{*}\hat\nu_{p-1}$
  \State $\Delta\hat\mu\gets -Z/Z'$;\quad clip $|\Delta\hat\mu|$ to $3\,\widehat{\mathrm{sd}}(x)$
  \State $\hat\mu\gets\hat\mu+\Delta\hat\mu$
  \If{$|\Delta\hat\mu|<\varepsilon\,\max(1,|\hat\mu|)$} \textbf{break} \EndIf
\EndFor
\State \Return $\hat\mu$
\end{algorithmic}
\end{algorithm}

For symmetric noise the iteration converges to the asymptotic variance
$1/(N\,\mathbf{b}^{\top}\mathbf{F}_2^{-1}\mathbf{b})$ predicted by
eq.~\eqref{eq:g2-alpha}; for asymmetric noise the signed moment
$\sigma_p\neq 0$ enters the score and the estimator retains a residual
bias of $O(\sigma_p)$ that does not vanish with~$N$. The latter is the
finite-asymmetry limitation of Form-B PATP discussed in
\S\ref{sec:patp:limits} and reported empirically in
\S\ref{sec:ill:are}; for the asymmetric Beta benchmark we therefore
recommend reporting the proxy estimator alongside the full one.

\subsection[Calibration of the optimal alpha]{Calibration of the optimal $\alpha^{*}$}\label{sec:alg:calibration}

\paragraph{Two-step procedure (basic version).}
The adaptive selection of $\alpha^{*}$ from~\eqref{eq:alpha-optimal} should
not rely only on classical cumulants in heavy-tail settings, because
$\hat\gamma_3$ and $\hat\gamma_4$ may be unstable or undefined. The basic
procedure is therefore formulated in terms of a preliminary center and a
chosen calibration criterion:

\begin{enumerate}
  \item \textbf{Preliminary estimation step.} Compute
        $\hat\theta^{(0)} = \hat\theta_{\mathrm{OLS}}$ and the corresponding
        residuals $\hat\xi_n = x_n - R(\hat\theta^{(0)}, z_n)$.
        Estimate either classical shape summaries
        $(\hat\gamma_3,\hat\gamma_4)$ when fourth moments are credible, or
        robust fractional summaries such as winsorized
        $\hat\nu_{p-1},\hat\nu_{p+1},\hat\nu_{2p},\hat\sigma_p$ on a grid
        of candidate $\alpha$ values.
  \item \textbf{Calibration step.} Select $\hat\alpha^{*}$ by one of three
        explicit criteria:
        (a) oracle minimization of theoretical $g_2(\alpha)$ when the
        distribution is known;
	        (b) plug-in minimization of $\hat g_2(\alpha)$ from estimated
	        fractional moments;
	        (c) lookup from an off-line table
	        $\alpha^{*}(\gamma_3,\gamma_4)$ as an engineering shortcut.
	        This third option is retained only in a non-theorem engineering
	        role and requires independent calibration data.
  \item \textbf{PATP estimation step.} Execute Newton-Raphson~\eqref{eq:nr-iteration}
        with fixed $\hat\alpha^{*}$ and initial value
        $\theta^{(0)} = \hat\theta^{(0)}$.
\end{enumerate}

\paragraph{Extended procedure with $k$ as a third parameter (finite-variance only).}
In cases where the two-parameter table $\alpha^{*}(\gamma_3, \gamma_4)$
yields an ambiguous result~--- in particular, for symmetric distributions with
$\hat\gamma_3 \approx 0$, where the excess $\hat\gamma_4$ alone cannot
distinguish structurally different shapes (Laplace vs.\ Simpson, see the
table in~\S\ref{sec:bg:entropy})~--- the extended calibration scheme with
the \emph{entropy coefficient}~$\hat k$ as a third parameter can be used
when the variance is finite and stably estimable. It is not a calibration
tool for infinite-variance regimes.
Algorithm:

\begin{enumerate}
  \item \textbf{Steps~1--2} as in the basic version: compute
        $\hat\theta^{(0)}, \hat\xi_n, \hat\gamma_3, \hat\gamma_4$.
  \item \textbf{Ambiguity check.} If the table
        $\alpha^{*}(\gamma_3, \gamma_4)$ at $(\hat\gamma_3, \hat\gamma_4)$
        yields a variance of the estimate $\alpha^{*}$ exceeding the threshold
        $\tau_{\alpha} = 0.1$~--- proceed to step~3'; otherwise
        continue with the standard procedure.
  \item[3'.] \textbf{Entropy coefficient estimation.}
        Via kernel density estimation (KDE) on $\hat\xi_n$ with an
        Epanechnikov or triangular kernel, compute
        $\hat H = -N^{-1}\sum_n \ln \hat f(\hat\xi_n)$ (an empirical
        approximation of the differential entropy by the plug-in method
        \citep[Ch.~12.4]{coverthomas2006}). Then
        $\hat k = e^{\hat H}/(2 \hat\sigma)$, where
        $\hat\sigma = \sqrt{m_2}$.
  \item[4'.] \textbf{Calibration with $k$.} Select $\hat\alpha^{*}$
        from the \emph{extended} table $\alpha^{*}(\gamma_3, \gamma_4, k)$
        by nearest-neighbor lookup on a three-dimensional grid.
  \item[5'.] \textbf{PATP estimation step} as in the basic version.
\end{enumerate}

\paragraph{When to use the extended procedure.}
The extended scheme is useful in the following practical scenarios:
(a)~symmetric distributions with $|\hat\gamma_3| < 0.1$ and
$|\hat\gamma_4| < 1$, where $\gamma_4$ is weakly informative;
(b)~suspected multimodal distributions, where $\gamma_4$ is masked by
component mixing;
(c)~finite-variance heavy-tail cases, where $\hat\gamma_4$ is unstable due
to the slow convergence of the sample fourth-order moment.
For typical biostatistical problems with $|\gamma_3| > 0.5$ the basic
two-parameter scheme is sufficient.

\paragraph{Remark on computing $\hat k$.}
The KDE estimate of entropy $\hat H$ has bias $O(N^{-2/(d+4)})$, where $d = 1$
for scalar $\xi$, i.e., $O(N^{-0.4})$ for $N \gtrsim 100$. This is
substantially slower convergence than the cumulant estimates
($O(N^{-1/2})$), so it should be treated as a diagnostic rather than a
theorem-level calibration statistic. For $N < 100$ it is recommended to use
only the basic scheme, since $\hat k$ may have large variance.

\paragraph{Infinite-variance calibration.}
When variance is not finite, both $\gamma_4$ and $k=e^H/(2\sigma)$ cease to
be appropriate calibration variables. In that setting one should use
quantile summaries, L-moments, trimmed L-moments, or a tail-index estimate
as shape diagnostics; a corresponding PATP calibration rule is left for
future work.

\paragraph{Alternative: grid search over $\alpha$.}
If the computational budget permits, instead of table-based calibration one
can perform a full grid search over $\alpha \in \{0, 0.05, 0.1,
\ldots, 0.45, 0.55, \ldots, 1\}$ (excluding the degenerate point $0.5$):
for each $\alpha$ compute the full PATP estimate and select $\hat\alpha^{*}$
by minimizing the empirical variance of the estimates (via bootstrap or
cross-validation). This approach is more expensive computationally but does
not depend on a pre-built table.

\paragraph{Bootstrap sensitivity.}
For any plug-in or table-based $\hat\alpha^*$, the final report should include
a sensitivity interval from bootstrap resampling of residuals. Instability
near $\alpha=1/2$ is expected and should be reported as a diagnostic rather
than hidden by rounding the selected value.

\paragraph{Convergence of calibration.}
The calibration of $\alpha$ is not a distribution-free theorem in the present
paper. In applications, $\alpha$ selection must be validated for the target
task and should not be transferred blindly between estimation, detection, and
classification problems. When the target task changes, task type should be
treated as an additional confounder in the calibration rule.

\subsection{Computational complexity and numerical aspects}\label{sec:alg:complexity}

\paragraph{Complexity per iteration.}
Each Newton-Raphson iteration~\eqref{eq:nr-iteration} requires:
\begin{itemize}
  \item computing $\xi_n$ and $\varphi_i(\xi_n; \alpha)$ for all
        $n \in \{1, \ldots, N\}$ and $i \in \{1, \ldots, S\}$:
        $\mathcal{O}(N S)$ operations;
  \item solving the normal equations~\eqref{eq:patp-normal-eq}
        of order $S \times S$: $\mathcal{O}(S^3)$ operations;
  \item evaluating $Z_N$ and $Z_N'$: $\mathcal{O}(N S)$ operations.
\end{itemize}
Total complexity per iteration: $\mathcal{O}(N S + S^3)$. For typical
$S = 2, 3$ and $N \gg S$ the first term dominates: $\mathcal{O}(N S)$.

\paragraph{Comparison with baseline location estimators.}
Table~\ref{tab:complexity} compares per-call complexity with the robust
baselines used in~\S\ref{sec:ill:baselines}. Wall-clock benchmarks at
$N=10^4$ from \texttt{R/results/runtime\_summary.csv} are shown in the
right-most column (Apple M-series laptop, single-threaded R~4.3 reference
implementation; rerun via \texttt{Rscript R/run\_all.R}).

\begin{table}[!htb]
\centering
\small
\begin{tabular}{l l l r}
\toprule
Estimator & Per-call complexity & Notes & $N=10^4$ \\
\midrule
OLS (sample mean) & $\mathcal{O}(N)$ & one pass & $0.02$\,ms \\
Median            & $\mathcal{O}(N\log N)$ & sort or quickselect & $0.09$\,ms \\
Huber location    & $\mathcal{O}(K_{\mathrm H}N)$ & iterative re-weighting & $1.09$\,ms \\
Median-of-means   & $\mathcal{O}(N)$ & $\sqrt{N}$ blocks $+$ median & $0.62$\,ms \\
PATP-2 proxy      & $\mathcal{O}(K_{\mathrm B}N)$ & bracketing root-find & $1.72$--$1.81$\,ms \\
PATP-2 full $\mathbf{F}_2^{-1}\mathbf{b}$ & $\mathcal{O}(K(NS+S^3))$ & one-step Newey (Alg.~\ref{alg:patp-full}) & $1.57$\,ms \\
\bottomrule
\end{tabular}
\caption{Per-call complexity and reference wall-clock at $N=10^4$ on the
shared Laplace benchmark. $K$ denotes the outer Newton iterations of the
PATP solver ($K=3$ default); $K_{\mathrm H}$ and $K_{\mathrm B}$ are the
inner iteration counts of Huber and bracketing, respectively. For PATP the
$S^3$ inversion is negligible at $S=2$, so cost scales linearly in~$N$.}
\label{tab:complexity}
\end{table}

\paragraph{Empirical convergence and conditioning.}
Across 200 Monte Carlo replicates on standardised Laplace ($N=100$), the
empirical condition number $\kappa(\mathbf{F}_2)$ at the OLS centre has
median $43$, $161$, $140$, $46$ for $\alpha\in\{0.05,0.30,0.70,0.95\}$
respectively, well below the safety cap $\kappa_{\max}=10^{10}$.
$\det\mathbf{F}_2$ collapses below $10^{-12}$ only inside the protected
band $|\alpha-\tfrac{1}{2}|<0.01$ where the implementation returns OLS by
design (Algorithm~\ref{alg:patp-full}, step~5). The median outer Newton
iteration count is $3$, with most replicates reaching $|\Delta\hat\mu|$
within one order of magnitude of $\varepsilon$ before the iteration cap;
in larger samples ($N\geq 500$) the empirical ARE for the symmetric
Laplace benchmark agrees with the closed-form $g_2(\alpha)$ within
$0.6$--$4.2\,\%$ across $\alpha\in\{0.05,0.30,0.70,0.95\}$ (Phase~B Gate
G2 in \texttt{R/run\_all.R}).

\paragraph{Complexity of $\alpha$ calibration.}
If a grid of $G$ candidate $\alpha$ values is evaluated, the leading cost is
\[
  \mathcal{O}\!\left(G(NS+S^3)\right)
\]
for fixed data. A bootstrap sensitivity check with $B$ resamples raises the
dominant scalar-location cost to approximately $\mathcal{O}(BGNS)$, plus the
cost of robust moment estimation and any baseline estimators.

\paragraph{Numerical considerations.}

\begin{itemize}
  \item \textbf{Small $|\xi_n|$.} The basis function $|\xi_n|^{p_i(\alpha)}$
        for $p_i(\alpha) < 1$ (e.g., $\alpha = 0$, $p_i = 1/i$)
        has an infinite derivative at zero. Numerically it is recommended to
        use a smoothed signed power or a bracketing method instead of
        relying on an unsafeguarded Newton step.
  \item \textbf{Degeneracy of the matrix $\mathbf{F}_S(\alpha)$.}
        As $\alpha \to 1/2$ the matrix approaches a degenerate
        state ($\det \mathbf{F}_S \to 0$, see~\S\ref{sec:eff:unimodal}).
        For $|\alpha - 1/2| < 0.05$ it is recommended to fall back to
        the OLS estimate directly.
  \item \textbf{Testing near $\alpha = 0$.} In the fractal
        regime $|\xi_n|^{1/i}$ grows slowly, which may require
        increasing the number of Newton-Raphson iterations. As a
        convergence safeguard~--- a damped scheme with factor $\lambda \in (0, 1]$:
        $\theta^{(k+1)} = \theta^{(k)} - \lambda \cdot Z_N/Z_N'$, with
        $\lambda = 0.5$ for the first $5$ iterations.
\end{itemize}

\paragraph{Implementation.}
A reference implementation in R accompanies this paper and is publicly
available at
\url{https://github.com/SZabolotnii/Ku-PATP-code-supplement}
under an MIT licence with \texttt{CITATION.cff} and
\texttt{REPRODUCIBILITY.md} included. The full
$\mathbf{F}_2^{-1}\mathbf{b}$ estimator from Algorithm~\ref{alg:patp-full}
is in \texttt{R/05\_full\_patp\_estimator.R}; the Monte Carlo driver,
robust baselines, and runtime benchmarks live in
\texttt{R/03\_monte\_carlo.R}. A production-grade port to the
\texttt{EstemPMM} CRAN package~\citep{zabolotnii2024nonlinear,
zabolotnii2026estempmmcran} is on the roadmap; the API mapping is
documented in Appendix~\ref{app:estempmm-roadmap}.

\subsection{Reproducibility}\label{sec:alg:repro}

\paragraph{Data.}
The simulation study uses synthetic samples from four symmetric laws
(Laplace, $\mathrm{GG}(1.5)$, $\mathrm{GG}(4)$, Student-$t(6)$); generators and
seed (\texttt{set.seed(2026)}) are at the top of \texttt{R/03\_monte\_carlo.R}.
The real-data examples use the \texttt{EuStockMarkets} data set distributed with
base \proglang{R} (\texttt{R/07\_real\_data\_application.R}) and two public-domain
daily FRED series --- the broad trade-weighted U.S.\ dollar index
(\texttt{DTWEXBGS}) and the CAD/USD spot rate (\texttt{DEXCAUS}) --- retrieved by
series id with a fixed SHA-256 checksum (\texttt{R/fetch\_fred.R};
\texttt{R/07b\_real\_data\_panel.R}).

\paragraph{Software stack.}
R~$\geq 4.3$, base \texttt{stats}, plus \texttt{dplyr}, \texttt{tidyr},
\texttt{ggplot2} for tabulation and figures. No external PATP package is
needed at runtime; all PMM routines are vendored in \texttt{R/}. The
Lean~4 formalization referenced in Appendix~\ref{app:lean-facts} builds
with Mathlib~v4.26.0 via \texttt{lake build} in the same repository.

\paragraph{One-line recipe.}
\begin{verbatim}
git clone https://github.com/SZabolotnii/Ku-PATP-code-supplement
cd Ku-PATP-code-supplement
Rscript R/run_all.R
\end{verbatim}
Total wall-clock $\approx 40$\,s on a single thread; the script regenerates
every CSV in \texttt{R/results/} and every PDF in \texttt{figures/} that
the manuscript references, including the new
\texttt{fig8\_proxy\_vs\_full.pdf} comparison.

\section{Illustrative examples}\label{sec:illustrations}

This section illustrates the theoretical results of~\S\ref{sec:efficiency} on
canonical distributions and on a real data set. The scope is
the estimation of a location parameter of a centred random variable~$\xi$ whose
distribution is symmetric (zero median); this is the regime in which the
closed-form $g_2(\alpha)$ is exact.

The numerical values below are generated by the accompanying \proglang{R}
pipeline. All PATP estimates use the full normal-equation estimator
$\hat\theta_{\PATP,S}(\alpha)$ of~\eqref{eq:patp-estimator}, built from the
$\mathbf{F}_S^{-1}\vec b$ weights (\S\ref{sec:alg:nr}). For asymmetric laws the Form-B basis
carries an $O(\sigma_p)$ residual term (\S\ref{sec:disc:robust}) and is out of
scope here.

\paragraph{Reproducibility note.}
All tables and figures in this section are regenerated by the accompanying
\proglang{R} pipeline, which produces the theoretical, Monte Carlo, convergence,
robust-baseline, alpha-ablation, regression and real-data summary outputs in the
public repository.

\subsection[Canonical distributions and theoretical g2(alpha)]{Canonical distributions and theoretical $g_2(\alpha)$}\label{sec:ill:canonical}

\paragraph{Laplace (double-sided exponential).}
Density $f(\xi) = \tfrac{1}{2}\,e^{-|\xi|}$; standardized cumulants
$\gamma_3 = 0$, $\gamma_4 = 3$ (heavy tails); fractional moments
$\nu_q = \Gamma(q+1)$ for $q > -1$. Substituting into~\eqref{eq:g2-alpha} with
$\sigma_p = 0$ gives
$g_2^{\mathrm{Lap}}(\alpha) = \frac{2\Gamma(2p+1) - \Gamma(p+2)^2}{2\bigl[\Gamma(2p+1) - 2p\,\Gamma(p+2)\Gamma(p) + 2p^2\Gamma(p)^2\bigr]}$,
with exact endpoints $g_2(0)\approx0.5425$ and $g_2(1)=3/4$; the minimum is at
the fractal boundary $\alpha=0$.

\paragraph{Generalized Gaussian $\mathrm{GG}(\beta)$.}
Density $f(\xi) \propto e^{-|\xi/\sigma|^{\beta}}$. The shape $\beta$ sets
$\gamma_4 = \Gamma(5/\beta)\Gamma(1/\beta)/\Gamma(3/\beta)^2 - 3$:
$\beta=1$ is Laplace, $\beta=2$ Gaussian, $\beta=4$ sub-Gaussian
($\gamma_4\approx-0.81$, platykurtic). The theoretical sweep moves the optimal
$\alpha$ from the fractal side for heavy tails to the signed-parity side for
platykurtic tails: $g_{2,\min}\approx0.913$ at $\alpha=0$ for $\mathrm{GG}(1.5)$,
and $g_{2,\min}\approx0.738$ at $\alpha=1$ for $\mathrm{GG}(4)$.

\paragraph{Student-$t$ with $\nu=6$ degrees of freedom.}
A symmetric heavy-tailed law with finite fourth moment for $\nu>4$
($\gamma_4 = 6/(\nu-4) = 3$ at $\nu=6$). Standardised to unit variance, its
absolute moments are
$\nu_q = \bigl(\tfrac{\nu-2}{\nu}\bigr)^{q/2}\nu^{q/2}
\Gamma(\tfrac{q+1}{2})\Gamma(\tfrac{\nu-q}{2})/[\Gamma(\tfrac12)\Gamma(\tfrac{\nu}{2})]$
for $-1<q<\nu$. Here $g_2(\alpha)$ has an \emph{interior} minimum,
$g_{2,\min}\approx0.862$ near $\alpha\approx0.62$, with endpoints
$g_2(0)\approx0.882$ and $g_2(1)\approx0.889$: neither boundary basis is
optimal, which is exactly the situation the continuous dial is designed for.

\paragraph{Cauchy (as a limiting example).}
Density $f(\xi)=\tfrac{1}{\pi(1+\xi^2)}$; moments $\nu_q$ finite only for $q<1$.
The classical PMM ($\alpha=1$, $p=2$) is inapplicable ($\nu_4=\infty$); the
fractal regime ($\alpha=0$, $p=1/2$) still needs the baseline second moment
$c_2$, which is infinite. Cauchy is therefore retained as a boundary
counterexample of infinite reference variance, not as a covered case.

\subsection[Single-task location: ARE vs. OLS]{Single-task location: ARE vs.\ OLS}\label{sec:ill:are}

\paragraph{Study design.}
Monte Carlo estimation of the location $\theta=\mu$ for four symmetric laws and
$\alpha \in \{0.05, 0.30, 0.70, 0.95\}$:
\begin{itemize}
  \item Distributions: Laplace ($\gamma_4=3$), $\mathrm{GG}(1.5)$ (moderately
        heavy), $\mathrm{GG}(4)$ ($\gamma_4<0$, light), Student-$t(6)$
        (heavy, $\gamma_4=3$).
  \item Sample sizes $N \in \{50,100,200,500\}$; $M=1000$ replications.
  \item Estimators: OLS (sample mean, baseline) and the full PATP estimator
        $\hat\theta_{\PATP,2}(\alpha)=\mathbf{F}_2^{-1}\vec b$
        (\texttt{R/03\_monte\_carlo.R}, using \texttt{R/05}).
  \item Metrics: empirical $\Var[\hat\mu_{\PATP}]$, bias, MSE, ARE.
\end{itemize}

\paragraph{Results.}
Fig.~\ref{fig:are_vs_N} reports $\mathrm{ARE}=\Var[\hat\mu_{\mathrm{OLS}}]/\Var[\hat\mu_{\PATP}]$.
For the full estimator $\mathrm{ARE}\ge1$ across \emph{all} distributions and all
$\alpha$ (PATP never below OLS), with the gain located where theory predicts:
Laplace peaks near the fractal side ($\mathrm{ARE}\approx1.75$ at $\alpha=0.05$,
$N=500$); $\mathrm{GG}(4)$ peaks near the signed-parity side
($\mathrm{ARE}\approx1.41$ at $\alpha=0.95$); Student-$t(6)$ peaks in the
interior. This is finite-sample evidence about the typical shape of
$g_2(\alpha)$, not a proof of general unimodality.

\begin{figure}[!htb]
  \centering
  \includegraphics[width=0.95\linewidth]{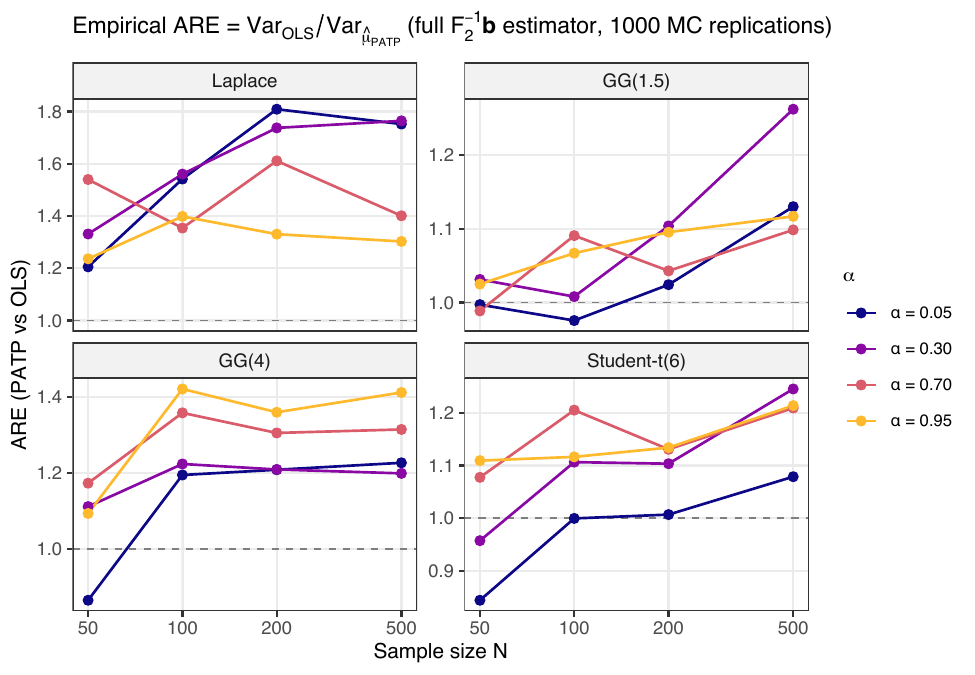}
  \caption{Empirical $\mathrm{ARE}=\Var[\hat\mu_{\mathrm{OLS}}]/\Var[\hat\mu_{\PATP}]$
    of the full $\mathbf{F}_2^{-1}\vec b$ estimator versus~$N$ (1000 MC
    replications). Points above the dashed line ($\mathrm{ARE}=1$) indicate PATP
    is more efficient than OLS; the gain shifts from the fractal side (Laplace)
    to the signed-parity side ($\mathrm{GG}(4)$) to the interior (Student-$t$),
    tracing the adaptive role of~$\alpha$.}
  \label{fig:are_vs_N}
\end{figure}

\subsection[Validation: convergence of empirical g2 to theory]{Validation: convergence of the empirical $g_2$ to the closed form}\label{sec:ill:validation}

The central claim of~\S\ref{sec:efficiency} is that $g_2(\alpha)$ is the
asymptotic variance ratio of the PATP estimator. Fig.~\ref{fig:validation}
validates this directly: for each symmetric law the empirical
$\hat g_2(\alpha)=\Var[\hat\mu_{\PATP}]/\Var[\hat\mu_{\mathrm{OLS}}]$ of the full
estimator is plotted against $N$ (up to $N=4000$, $M=5000$ replications) at the
fractal ($\alpha=0.05$) and signed-parity ($\alpha=0.95$) endpoints, alongside
the closed-form value (dashed). The empirical ratios converge onto the
theoretical lines: at $N=4000$ the gap is $0.9$--$1.1\%$ for Laplace and
$0.7$--$1.7\%$ for $\mathrm{GG}(1.5)$, and within a few per cent for
$\mathrm{GG}(4)$ and Student-$t(6)$ (each bootstrap point carries a Monte Carlo
standard error of order $2\%$).

\begin{figure}[!htb]
  \centering
  \includegraphics[width=0.95\linewidth]{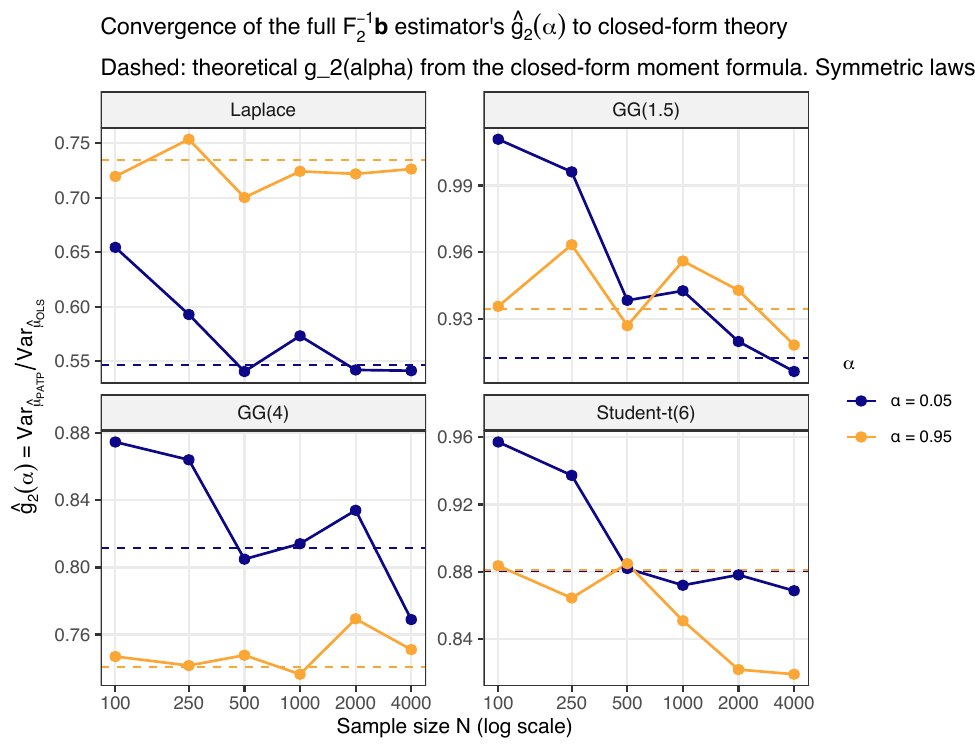}
  \caption{Convergence of the full $\mathbf{F}_2^{-1}\vec b$ estimator's
    empirical $\hat g_2(\alpha)$ to the closed-form $g_2(\alpha)$
    (dashed) as $N$ grows, for the four symmetric laws at $\alpha\in\{0.05,0.95\}$
    ($M=5000$). The empirical ratios track the theoretical values, confirming
    that the closed form is the finite-sample efficiency of the estimator.}
  \label{fig:validation}
\end{figure}

To show why the full estimator is required, Fig.~\ref{fig:proxy_vs_full}
contrasts it with the scalar signed-power M-estimator that solves only the
single equation $\E[\mathrm{sign}(\xi)|\xi|^p]=0$: the latter departs sharply
from the closed form on the signed-parity side (e.g.\ for Laplace at
$\alpha=0.95$ its empirical $g_2$ exceeds $2.5$, indicating a loss against OLS),
whereas the full normal-equation estimator follows the theoretical curve.

\begin{figure}[!htb]
  \centering
  \includegraphics[width=0.9\linewidth]{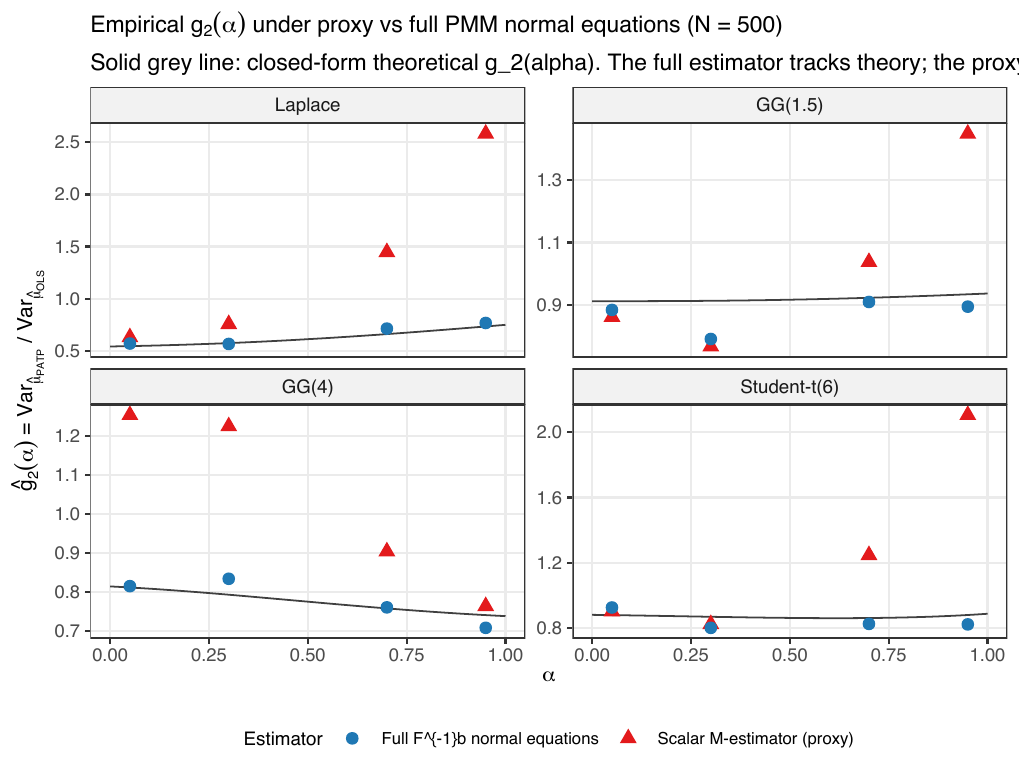}
  \caption{Empirical $\hat g_2(\alpha)$ at $N=500$ under the scalar M-estimator
    proxy (red) and the full $\mathbf{F}_2^{-1}\vec b$ normal-equation estimator
    (blue), against the closed-form $g_2(\alpha)$ (grey). Only the full estimator
    tracks the theory; the scalar proxy diverges on the signed-parity side.}
  \label{fig:proxy_vs_full}
\end{figure}

\paragraph{Bias.}
Fig.~\ref{fig:bias_var} confirms the estimator is practically unbiased on these
symmetric laws ($|\mathrm{bias}|$ one-to-two orders below $\sqrt{\Var}$), so the
efficiency change is driven by variance, not bias.

\begin{figure}[!htb]
  \centering
  \includegraphics[width=0.95\linewidth]{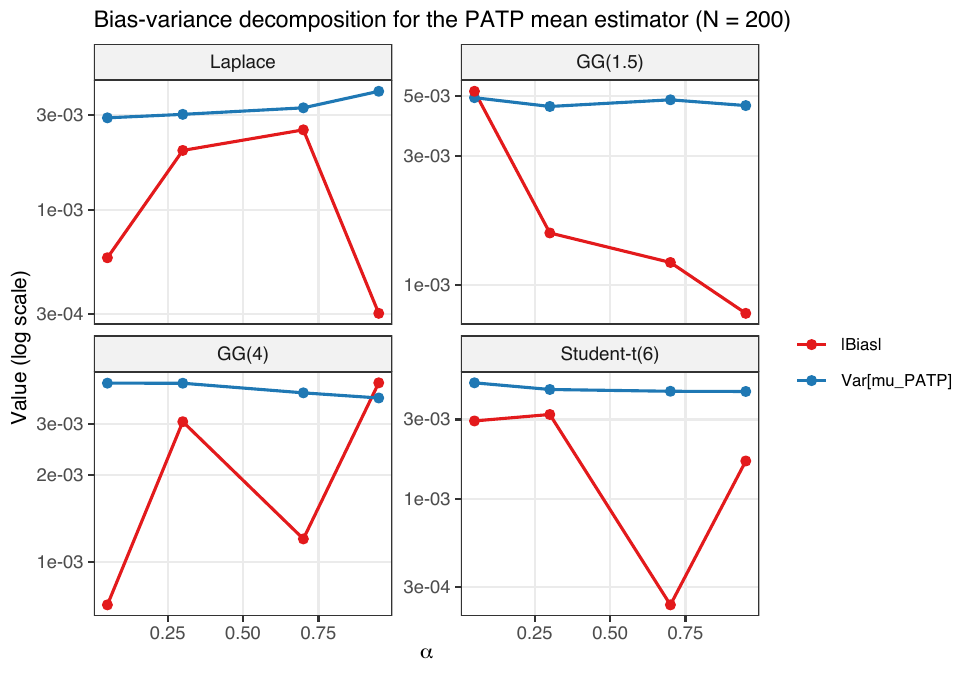}
  \caption{Bias-variance decomposition of the full PATP estimator ($N=200$):
    variance (blue) and absolute bias (red) versus $\alpha$. In all regimes
    $|\mathrm{bias}|\ll\sqrt{\Var}$.}
  \label{fig:bias_var}
\end{figure}

\subsection{Robust baselines and alpha ablation}\label{sec:ill:baselines}

\paragraph{Baselines.}
On the same Monte Carlo design the pipeline also computes six scalar location
baselines (sample mean, median, $10\%$ trimmed and winsorized means, Huber
location, median-of-means). For
Laplace at $N=100$ the median attains relative MSE $\approx0.53$ against the
sample mean and the full PATP estimator at $\alpha=0.05$ attains
$\hat g_2\approx0.61$, so PATP is competitive with the strongest classical robust
estimators while retaining a closed-form efficiency. For $\mathrm{GG}(4)$ the
mean is near-best, consistent with light tails and warning against any universal
superiority claim.

\paragraph{Alpha ablation.}
The grid ablation shows the
expected directional pattern: Laplace selects the fractal side, $\mathrm{GG}(4)$
the signed-parity side, and Student-$t(6)$ the interior, supporting the use of
$\alpha$ as an adaptive tuning parameter selected by the grid criterion of
\S\ref{sec:alg:calibration}.

\paragraph{Runtime.}
The lightweight benchmark confirms the
$\mathcal{O}(N)$ per-iteration cost: at $N=10^5$ the full estimator runs in
$\approx1.1\times10^{-2}$\,s per fit, the same order as Huber location.

\subsection{Regression coefficients}\label{sec:ill:regression}

Because $g_2(\alpha)$ multiplies the OLS coefficient covariance through the same
scalar factor, the variance reduction carries over from location to linear
regression. A Monte Carlo check on $y=\beta_0+\beta_1 x+\xi$ with symmetric
errors ($n=400$, $M=2000$) confirms that the
empirical slope-variance ratio $\Var[\hat\beta_{1,\PATP}]/\Var[\hat\beta_{1,\mathrm{OLS}}]$
matches the closed-form $g_2(\alpha)$ to within a few per cent, with the
degenerate value $g_2(1/2)=1$ recovered exactly (table omitted). For example, for $\mathrm{GG}(4)$ errors at
$\alpha=0.95$ the empirical ratio is $0.763$ against the predicted $0.741$.

\subsection{Application to real data}\label{sec:ill:realdata}

We demonstrate the estimator on real, non-simulated data: daily log-returns of
the four indices in the \texttt{EuStockMarkets} data set (1991--1998,
$N=1859$). Equity-index returns are leptokurtic and approximately symmetric, so
they fall in the regime of the symmetric $g_2(\alpha)$. We select the most
symmetric, leptokurtic series (FTSE: skewness $0.11$, excess kurtosis $2.64$)
and estimate its location. The sample is first standardised to unit robust scale
($z=(x-\mathrm{med})/\mathrm{MAD}$); since $g_2$ is scale-invariant and location
estimation is scale-equivariant this does not change the target, and it places
the data in the unit-scale regime in which the estimator was validated.

Following \S\ref{sec:alg:calibration}, $\alpha^{*}$ is chosen by the robust grid
criterion --- the value minimising the bootstrap variance of the estimator
($B=2000$). The realised efficiency
$\hat g_2(\alpha)=\Var[\hat\mu_{\PATP}]/\Var[\hat\mu_{\mathrm{mean}}]$ is a smooth,
convex curve over the grid with an interior minimum at $\alpha^{*}=0.55$,
\[
  \hat g_2(\alpha^{*}) = 0.868
  \quad\Longrightarrow\quad 13.2\%\ \text{variance reduction vs.\ the sample mean.}
\]
The result agrees with theory: matching the standardised excess kurtosis
($2.64$) to a Student-$t$ gives $\nu\approx6.3$, whose closed-form
$g_2(\alpha^{*})=0.874$ differs from the realised value by $0.7\%$. PATP is also
competitive with classical robust baselines on the same data (median
$\hat g_2=0.859$, Huber $\hat g_2=0.888$), while uniquely providing an analytic
efficiency and an interpretable selected exponent.

Fig.~\ref{fig:real_data_bootstrap} contrasts the bootstrap distributions of the
competing location estimators: the PATP estimator at $\alpha^{*}$ is markedly
tighter than the sample mean and on par with the robust baselines, the visual
counterpart of the $13.2\%$ variance reduction. That this gain is modest --- and
$\alpha^{*}$ only weakly identified --- reflects the near-Gaussian tails of equity
returns, a shallowness made explicit across heavier-tailed series below
(Fig.~\ref{fig:real_data_panel}). The numbers and figure are reproducible from the
accompanying \proglang{R} pipeline.

\begin{figure}[!htb]
  \centering
  \includegraphics[width=0.95\linewidth]{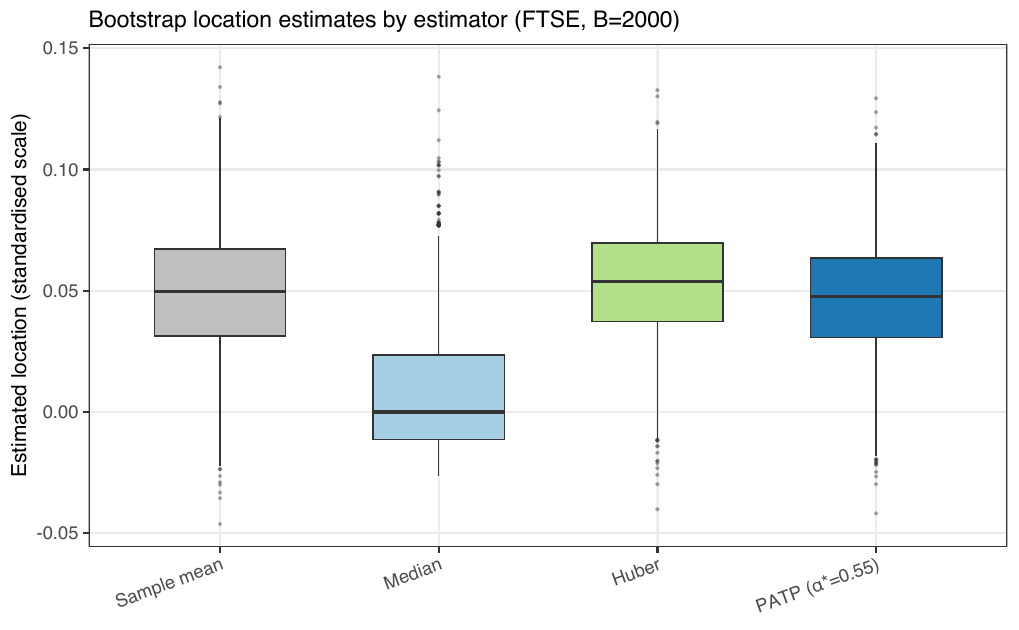}
  \caption{Bootstrap distributions ($B=2000$) of the competing location
    estimators for the FTSE daily log-return location (\texttt{EuStockMarkets},
    standardised to unit robust scale); the PATP estimator at $\alpha^{*}=0.55$
    achieves a spread comparable to the robust baselines and well below the
    sample mean.}
  \label{fig:real_data_bootstrap}
\end{figure}

\paragraph{The adaptive role of $\alpha$ across tail weights.}
The shallow FTSE surface is itself diagnostic: equity-index returns are only
mildly leptokurtic ($\gamma_4\approx2.6$), so no basis is far from optimal and
the gain is modest. To exhibit the adaptive role of $\alpha$ directly we apply
the \emph{same} estimator to a panel of symmetric real series of increasing tail
weight (all $|\hat\gamma_3|<0.12$): the FTSE returns above; daily log-changes of
the broad trade-weighted U.S.\ dollar index (FRED \texttt{DTWEXBGS}; $N=5125$,
$\gamma_4\approx4.3$); and daily log-returns of the CAD/USD spot rate (FRED
\texttt{DEXCAUS}, 1971--2026; $N=13{,}909$, $\gamma_4\approx9.3$) --- a
commodity-currency series that is heavy-tailed yet, unusually, almost symmetric
($\hat\gamma_3=-0.08$). As the tail weight rises, the data-driven $\alpha^{*}$
slides monotonically from the signed-parity side toward the fractal endpoint
($0.55\to0.40\to0.05$) and the variance reduction over the sample mean deepens
($13.2\%\to28.2\%\to53.9\%$); see Fig.~\ref{fig:real_data_panel}.

On the heaviest series the fixed classical basis ($\alpha=1$) carries
$1.57\times$ the variance of the adaptive optimum ($\hat g_2=0.724$ vs.\ $0.462$):
adapting $\alpha$ --- here all the way into the fractal regime $\alpha^{*}=0.05$,
inaccessible to the integer-power PMM --- recovers efficiency that no fixed basis
attains, and PATP also overtakes the Huber baseline as the tails grow heavier
($\hat g_2=0.462$ vs.\ $0.600$ on CAD/USD). The selected exponents trace an
empirical $\alpha^{*}(\gamma_4)$ consistent with the kurtosis-matched Student-$t$
prediction (dotted curves in Fig.~\ref{fig:real_data_panel}), the data-side
counterpart of the topographic map of \S\ref{sec:ill:topographic}. The trend is
robust to the treatment of unchanged-quote days: excluding zero-return
observations, CAD/USD still attains a $39\%$ reduction at $\alpha^{*}=0.45$. (On
such price-quantised series the sample median is not a stable comparator --- its
bootstrap variance swings from $0.31\times$ to over $4\times$ the sample mean's
under that change --- so Huber is used as the robust baseline.)

\begin{figure}[!htb]
  \centering
  \includegraphics[width=0.95\linewidth]{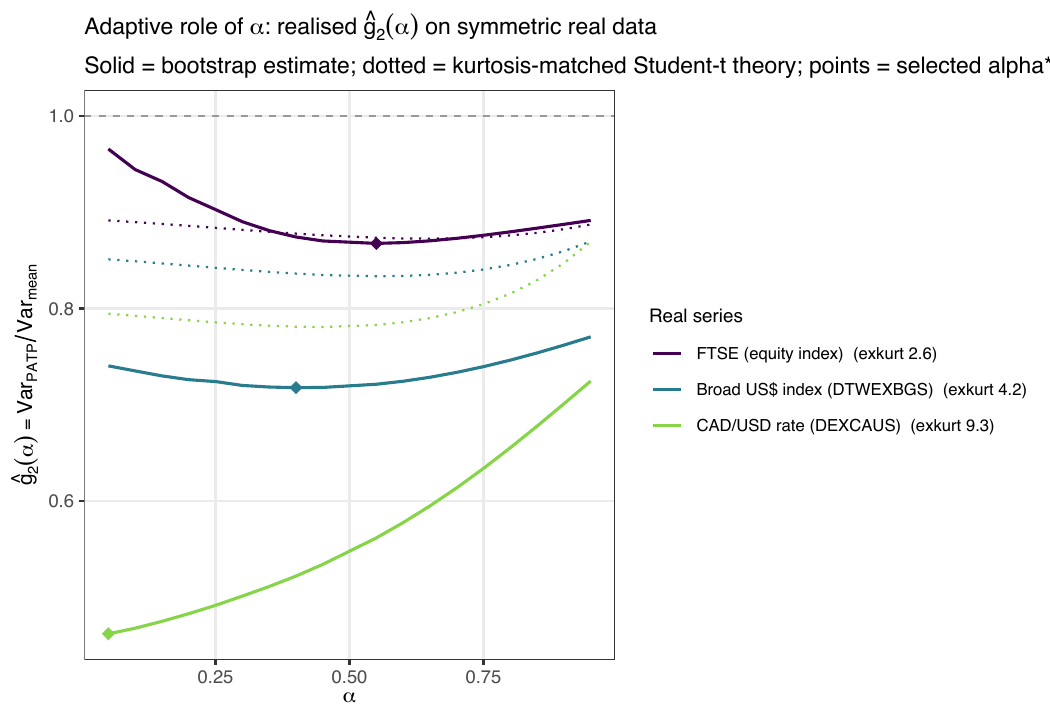}
  \caption{Adaptive role of $\alpha$ on symmetric real data of increasing tail
    weight. Each solid curve is the realised
    $\hat g_2(\alpha)=\Var[\hat\mu_{\PATP}]/\Var[\hat\mu_{\mathrm{mean}}]$
    (bootstrap, $B=2000$); dotted curves are the kurtosis-matched Student-$t$
    theory; markers show the selected $\alpha^{*}$. As $\gamma_4$ grows
    ($2.6\to4.3\to9.3$) the curve deepens and its minimum slides from the
    signed-parity side ($\alpha^{*}=0.55$, FTSE) to the fractal endpoint
    ($\alpha^{*}=0.05$, CAD/USD).}
  \label{fig:real_data_panel}
\end{figure}

\subsection[Topographic plane (kappa, k)]{Topographic plane $(\varkappa, k)$}\label{sec:ill:topographic}

The topographic classification of~\S\ref{sec:bg:entropy} places the canonical
laws on the plane $(\varkappa, k)$ of contrexcess and entropy coefficient.
Fig.~\ref{fig:topographic} (in~\S\ref{sec:bg:entropy}) shows the
$\mathrm{GG}(\beta)$ curve and the canonical points; the symmetric laws of this
section lie along it. The advantage of the two-dimensional representation over
$\alpha^{*}(\gamma_4)$ alone is that it resolves the shape-identification
ambiguity discussed in~\S\ref{sec:bg:entropy}: distributions that share a value
of $\gamma_4$ but differ in tail structure (e.g.\ Laplace and the triangular
law) separate in $k$, which tracks the optimal~$\alpha^{*}$ more faithfully. A
full empirical map of $\alpha^{*}(\varkappa,k)$ across a wider population of
distributions is a natural direction for the cross-domain study
(\S\ref{sec:disc:future}).

\section{Discussion}\label{sec:discussion}

\subsection{Connections to neighboring semiparametric frameworks}\label{sec:disc:related}

\paragraph{PATP $\leftrightarrow$ GMM.}
Hansen's Generalized Method of Moments~\citep{hansen1982gmm} estimates a
parameter through a system of moment conditions
$\E[g(X; \theta)] = 0$ with a weighting matrix $W$ optimized to minimize
asymptotic variance. PATP can be interpreted as a \emph{special form of GMM},
where the moment functions are PATP basis functions
$g_i(X; \theta, \alpha) = \varphi_i(X - R(\theta); \alpha)$, and the
weighting matrix is replaced by the centered correlant matrix
$\mathbf{F}_S(\theta; \alpha)$. A formal proof of this equivalence remains
an open problem; in particular, whether the PATP weighting matrix
$\mathbf{F}_S(\alpha)$ is optimal under Hansen's criterion for the
corresponding set of moment conditions is a question requiring separate
analysis.

\paragraph{PATP $\leftrightarrow$ M-estimators.}
Huber's M-estimators~\citep{huber1981robust} are defined through
$\psi$-functions optimized for robustness to outliers. The PATP estimator
\eqref{eq:patp-estimator} is an $M$-estimator with $\psi$-function
$\psi_{\PATP}(\xi; \alpha) = \sum_i h_i^{*}(\alpha) \cdot
\partial_{\xi}\varphi_i(\xi; \alpha)$, i.e., a linear combination of
even-order derivatives of the basis functions with weight coefficients
depending on the estimated cumulants. The distinction of PATP from
classical Huber or Tukey M-estimators is the automatic selection of the
$\psi$-function from the estimated characteristics of the noise distribution,
rather than an a priori choice of regularization parameters.

\paragraph{PATP $\leftrightarrow$ L-moments.}
Hosking's L-moments~\citep{hosking1990lmoments} are linear combinations of
order statistics and offer an alternative parameterization of distributions
that is robust to outliers. Conceptually close to PATP in logic~--- using
moment-like characteristics without assuming finiteness of classical
moments~--- L-moments are realized through order statistics, while PATP
uses fractional absolute moments. A formal bridge between these two
frameworks does not yet exist in the literature; constructing one is a
promising topic for a separate paper.

\paragraph{PATP $\leftrightarrow$ SLS.}
The Second-order Least Squares method
(SLS)~\citep{wang2008sls} augments the OLS criterion with a quadratic term
from second-moment deviations, yielding a strict gain over OLS for
asymmetric distributions. The quantitative equivalence of PMM2
(i.e., the classical even-power PMM2 case, not the Form-B signed-parity
endpoint used in this paper) and SLS
was demonstrated in~\citep{zabolotnii2024nonlinear} for nonlinear
regression problems with $\chi^2$ errors. PATP should therefore be read as a
neighboring continuous fractional-power construction rather than as a literal
SLS generalization. The extension to $\alpha < 1/2$ (the fractal branch) is a
region inaccessible to classical SLS.

\subsection{Robust statistics for heavy tails}\label{sec:disc:robust}

\paragraph{Position of PATP.}
PATP fits into the spectrum of semiparametric methods between classical PMM
of Kunchenko (which requires finite high-order moments) and nonparametric
robust methods (median, MAD, which are trimmed at outliers but ignore the
distributional structure). The advantage of PATP is the closed-form
$g_2(\alpha)$, which allows \emph{analytical} efficiency estimation for
each fixed $\alpha$ and \emph{adaptive} calibration to the estimated noise
profile. Bootstrap remains useful for sensitivity analysis of the selected
$\alpha$, especially near the singular point $\alpha=1/2$.
This is a model- and moment-structured goal, distinct from Catoni,
median-of-means, and related estimators whose main target is
distribution-free concentration under weak assumptions.

\paragraph{Integrability constraints.}
All results in the present paper assume finiteness of the moments required by
\S\ref{sec:eff:s2}. For $\alpha \in [0, 1]$ and $i = 2$ the highest
nonlinear exponent is $2p_2(\alpha) \leq 4$ (attained at $\alpha = 1$).
However, because the present $S=2$ construction keeps the linear basis
function $\varphi_1(\xi)=\xi$, the matrix $\mathbf F_2$ also contains
$F_{11}=c_2$. Thus PATP$(S=2)$ relaxes the classical fourth-moment
requirement to a finite-variance requirement near $\alpha=0$; it does not
provide an OLS-referenced variance formula for infinite-variance laws such
as Cauchy or stable distributions with $\alpha_s<2$.

For such laws, a different target location functional, bounded or
kernel-weighted moments, or a modified PATP basis without the unsmoothed
linear component is required. Analytical investigation of these variants is
beyond the scope of the present paper and remains a direction for future
work.

\paragraph{Honest acknowledgment of unproven claims.}
The unimodality of $g_2(\alpha)$ as a function of $\alpha$ (for typical
profiles $\gamma_3, \gamma_4$) is not proved here. The figures in this paper
are diagnostic evidence about typical curve shapes, not a replacement for a
general theorem. A formal proof would require analysis of a rational function
of several integral characteristics.

\paragraph{Scope.}
The present manuscript establishes the Form-B PATP basis, its moment
conditions, the $S=2$ efficiency formula $g_2(\alpha)$, and a reproducible
simulation study (\S\ref{sec:illustrations}) that validates the formula for
symmetric error laws using the full $\mathbf{F}_S^{-1}\vec b$ normal-equation
estimator. It does not claim a distribution-free robust estimator for arbitrary
(including infinite-variance) laws; the asymmetric Form-B residual term and the
infinite-variance boundary are deferred to future work.

\subsection{Future directions}\label{sec:disc:future}

\paragraph{Cross-domain transfer.}
Empirical verification of the cross-domain transferability of $\alpha^{*}$
across detection, classification, and denoising examples remains future work.
The likely task-confounder hypothesis
$\alpha^{*} = f(\gamma_3, \gamma_4, \mathrm{task})$ points to the need for
further extension of the theoretical model.

\paragraph{Applied implementation in FP-regression.}
The discrete $\alpha = 0$ implementation of PATP can be combined with the
Royston--Altman basis for fractional polynomial regression~\citep{royston1994}.
Developing this applied regression variant is left as future work; the present
paper focuses on the basis construction and the $S=2$ efficiency formula.

\paragraph{Lean 4 formalization.}
The limiting cases of the exponents $p_i(\alpha)$ and the basis
$\varphi_i(\xi; \alpha)$ have been formalized in Lean~4 with Mathlib v4.26.0
across five small modules; see Appendix~\ref{app:lean-facts} for theorem
names. This formal layer is intentionally algebraic. Moment existence,
$L_2(P)$ positive definiteness, and asymptotic distribution theory remain
conventional mathematical assumptions and proof obligations.

\paragraph{Extension of the PATP family.}
A separate direction is the introduction of a \emph{continuously vectorial}
$\vec\alpha \in [0,1]^{S-1}$ in place of the scalar $\alpha$: each basis
exponent $p_i$ is governed by its own parameter $\alpha_i$. This increases
adaptivity at the cost of overfitting risk for small samples. Numerical
comparison of the scalar and vector PATP families for a broad class of
biostatistical and financial data is the subject of a separate study.

\section{Conclusions}\label{sec:conclusion}

The Parametrically Adaptive Transition Polynomial (PATP) is a continuously
$\alpha$-parameterized family of sign-preserving fractional-power basis
functions that fits into Kunchenko's generalized framework of stochastic
polynomials as a continuous interpolator between three structural points: the
fractal regime ($\alpha = 0$, $p_i = 1/i$), the degenerate linear regime
($\alpha = 1/2$, all $p_i = 1$), and the signed-parity integer-polynomial
regime ($\alpha = 1$, $p_i = i$).

\paragraph{Main Theoretical Results.}
\begin{enumerate}
  \item \textbf{Formal definition of the PATP family} (\S\ref{sec:patp})
        with a quadratic Lagrangian parameterization
        $p_i(\alpha) = 1/i + (4 - i - 3/i)\alpha + (2i - 4 + 2/i)\alpha^2$,
        machine-verified in Lean 4 (modules \texttt{PATP.Param},
        \texttt{PATP.Basis}).

  \item \textbf{Closed-form formula for $g_2(\alpha)$} (\S\ref{sec:eff:s2})
        for the variance reduction coefficient of the degree $S=2$ PATP
        estimate in the mean estimation problem, expressed in terms of
        four fractional moments $\nu_{p-1}, \nu_{p+1}, \nu_{2p}, \sigma_p$
        with a single controlling exponent $p = p_2(\alpha)$.

  \item \textbf{Degeneracy and limiting cases} (\S\ref{sec:eff:unimodal}):
        The matrix $\mathbf{F}_2(\alpha = 1/2)$ is strictly singular;
        the signed-parity integer-power regime ($\alpha = 1$) and the
        fractal regime ($\alpha = 0$) yield structurally different
        $g_2$ formulas, with different integrability requirements
        (a fourth moment for the classical reference formula versus finite
        variance plus fractional nonlinear moments near $\alpha=0$).

  \item \textbf{Structural expansion of integrability.} In the fractal
        regime $\alpha = 0$, the present $S=2$ formula lowers the nonlinear
        requirement from fourth-order moments to fractional orders up to
        $3/2$, while retaining the finite-variance requirement caused by
        the linear basis component.
\end{enumerate}

\paragraph{Empirical validation.}
A reproducible \proglang{R} study confirms the theory: the full
$\mathbf{F}_2^{-1}\vec b$ estimator's Monte Carlo $g_2(\alpha)$ converges to the
closed form across the symmetric canonical laws (within $1$--$2\%$ at $N=4000$
for Laplace and $\mathrm{GG}(1.5)$; \S\ref{sec:ill:validation}), the same factor
governs the regression-slope variance, and a real-data location example (FTSE
daily log-returns) attains a $13\%$ variance reduction over the sample mean,
matching the kurtosis-matched Student-$t$ prediction to within $0.7\%$.

\paragraph{Practical Recommendations.}
\begin{itemize}
  \item For symmetric heavy-tailed errors (Laplace, $\mathrm{GG}(\beta<2)$,
        Student-$t$): the fractal branch $\hat\alpha^{*} \in [0, 0.5)$ is the
        relevant regime; on Laplace it attains $\mathrm{ARE}\approx1.75$ over OLS.
  \item For symmetric light-tailed (platykurtic) errors ($\mathrm{GG}(\beta>2)$):
        the signed-parity branch $\hat\alpha^{*} \in (0.5, 1]$ is preferred
        ($\mathrm{ARE}\approx1.4$ for $\mathrm{GG}(4)$ at $\alpha=1$).
  \item For moderate symmetric leptokurtosis the optimum is interior; select
        $\hat\alpha^{*}$ by the grid-plus-bootstrap criterion of
        \S\ref{sec:alg:calibration}, as in the FTSE real-data example
        ($\hat\alpha^{*}=0.55$, $13\%$ variance reduction).
  \item For asymmetric errors the Form-B basis carries an $O(\sigma_p)$ residual
        term and lies outside the present scope; for infinite-variance laws
        (Cauchy, stable $\alpha_s<2$) the OLS-referenced $g_2(\alpha)$ does not
        apply (future work).
\end{itemize}

\paragraph{Limitations.}
The formula for $g_2(\alpha)$ proven in \S\ref{sec:eff:s2} requires:
(a) a specified location convention for the residual $\xi=X-\theta_0$;
(b) the finiteness of $c_2$ and fractional moments
$\nu_{p-1}, \nu_{p+1}, \nu_{2p}$
for a given $p = p_2(\alpha)$;
(c) no point mass at zero and mild local regularity near zero for
the derivative of the signed-parity function to be integrable.
The simulation study and the real-data example use the full
$\mathbf{F}_2^{-1}\vec b$ normal-equation estimator, whose empirical $g_2(\alpha)$
converges to the closed form (\S\ref{sec:ill:validation}); the scope is limited
to symmetric error laws, since for asymmetric laws the Form-B basis retains an
$O(\sigma_p)$ bias.

Typical curve shapes of $g_2(\alpha)$ have been explored numerically in this
paper, but a formal unimodality theorem is not claimed in the present
manuscript.

\paragraph{Perspectives.}
The PATP as a methodological basis opens several directions:
applied implementation in fractional-polynomial regression;
formalization of the remaining $g_2(\alpha)$ algebra in Lean 4;
measure-theoretic and asymptotic formalization beyond the present Lean
algebraic layer; vectorial expansion $\vec\alpha \in [0,1]^{S-1}$; integration with
L-moments and robust statistics for heavy tails; extensive
cross-domain testing with task-conditioning. The current paper establishes
the foundation for all these directions and identifies the limitations for each.

\paragraph{Acknowledgments.}
The author expresses gratitude to the scientific school of Yuriy Petrovych
Kunchenko (Cherkasy State Technological University), whose ideas form the
basis of the current work.

\appendix
\section{Technical Notes Kept Outside the Main Line}\label{app:technical}

\subsection{Local integrability at the origin}\label{app:local-integrability}

The derivative of the signed-parity basis
$\varphi(\xi)=\mathrm{sign}(\xi)|\xi|^p$ is
$p|\xi|^{p-1}$ for $\xi\neq0$.  For $p<1$ this derivative is unbounded at
zero, but unboundedness is not by itself a divergence of the expectation.
If the distribution has no atom at zero and its density is bounded by $M$ on
$[-\varepsilon,\varepsilon]$, then for every $p>0$
\[
  \E\!\left[|\xi|^{p-1}\mathbf{1}_{|\xi|\leq\varepsilon}\right]
  \leq
  2M\int_0^\varepsilon x^{p-1}\,dx
  =
  \frac{2M}{p}\varepsilon^p
  <\infty .
\]
Thus the fractal endpoint of the $S=2$ construction, where
$p=p_2(0)=1/2$, is locally integrable under ordinary bounded-density
regularity near the center.  The real obstruction for extremely heavy-tailed
laws is the tail condition in the moments $\nu_{p+1}$ and $\nu_{2p}$, not a
generic singularity at zero.  Atoms or singular densities at the center still
require smoothing or a bracketing solver, as described in
\S\ref{sec:alg:nr}.

\subsection{Lean-verified algebraic layer}\label{app:lean-facts}

Table~\ref{tab:lean-structural} records the formalized part of the PATP
construction.  The purpose of this Lean layer is deliberately narrow: it
checks the exponent identities, the signed-parity basis identities, the
midpoint collapse, and the final algebraic step in the displayed
$g_2(\alpha)$ formula.  It does not claim to formalize moment existence,
asymptotic normality, or distribution-free robustness.

\begin{table}[!htb]
\centering
\small
\setlength{\tabcolsep}{3pt}
\begin{tabular}{p{0.34\linewidth}p{0.47\linewidth}p{0.12\linewidth}}
\hline
Structural fact & Lean module / theorem & Scope \\
\hline
$p_i(0)=1/i$, $p_i(1/2)=1$, $p_i(1)=i$
  & \texttt{PATP.Param}: boundary lemmas
  & Algebraic \\
$\varphi_i(\xi;1/2)=\xi$ and oddness of Form-B
  & \texttt{PATP.Basis}: midpoint and oddness lemmas
  & Basis \\
$p_i(\alpha)-p_j(\alpha)$ factorization
  & \texttt{PATP.ExponentSeparation}: factorization lemma
  & Algebraic \\
Pairwise basis collapse at $\alpha=1/2$
  & \texttt{PATP.Degeneracy}: pairwise collapse lemma
  & Basis \\
Final algebraic step in $g_2(\alpha)$
  & \texttt{PATP.G2Algebra}: \texttt{g2Formula\_eq\_from\_bFb}
  & Algebraic \\
\hline
\end{tabular}
\caption{Lean-verified structural facts used in the PATP definition and
the $S=2$ efficiency formula.  Probability-theoretic assumptions remain in
the manuscript, because they concern the underlying distribution rather than
pure algebra.}
\label{tab:lean-structural}
\end{table}

\subsection{EstemPMM integration roadmap}\label{app:estempmm-roadmap}

The reference R pipeline at
\url{https://github.com/SZabolotnii/Ku-PATP-code-supplement}
is structured as a self-contained replication artefact for this paper;
its long-term home is the \texttt{EstemPMM} package on CRAN
(currently shipping the integer-power PMM2 and PMM3 estimators,
\citealp{zabolotnii2026estempmmcran}). Table~\ref{tab:estempmm-roadmap}
maps the prototype functions in \texttt{R/} to the planned API in the
next \texttt{EstemPMM} release. The aim is to graduate the PATP
prototype without rewriting any of the closed-form algebra of
\S\ref{sec:efficiency}: the full $\mathbf{F}_2^{-1}\mathbf{b}$ solver of
Algorithm~\ref{alg:patp-full} becomes a single user-facing function, and
the diagnostics in \S\ref{sec:alg:repro} become standard \texttt{summary}
methods on the returned fit object.

\begin{table}[!htb]
\centering
\small
\setlength{\tabcolsep}{3pt}
\begin{tabular}{p{0.39\linewidth}p{0.55\linewidth}}
\hline
Prototype function (this paper) & Planned \texttt{EstemPMM} API \\
\hline
\texttt{p\_i(i, alpha)}
  & \texttt{patp\_exponent(i, alpha)} (utility) \\
\texttt{empirical\_moments(x, mu, p)}
  & internal helper of \texttt{patp\_moments(fit)} \\
\texttt{build\_F2\_b\_hstar(x, mu, p)}
  & \texttt{patp\_correlant(x, alpha, mu)} returning $\mathbf{F}_2,\mathbf{b},\mathbf{h}^{*}$ \\
\texttt{patp\_full\_estimator(x, alpha, \ldots)}
  & \texttt{lm\_patp(formula, data, alpha, control)} (regression-aware)
    and \texttt{patp\_location(x, alpha)} (scalar) \\
\texttt{patp\_proxy\_internal(x, alpha, \ldots)}
  & \texttt{patp\_location(x, alpha, method = "proxy")} \\
\texttt{g2\_alpha(alpha, c2, nu, sigma)}
  & \texttt{patp\_g2(alpha, distribution)} closed-form ARE evaluator \\
Robust baselines block in \texttt{R/03\_monte\_carlo.R}
  & reused via \texttt{robustbase}; not re-exported \\
\hline
\end{tabular}
\caption{Mapping from the manuscript's reference R pipeline to the
planned PATP API in \texttt{EstemPMM}. The user-facing entry points
(\texttt{lm\_patp}, \texttt{patp\_location}, \texttt{patp\_g2}) mirror
the existing \texttt{lm\_pmm2} / \texttt{lm\_pmm3} naming convention of
the CRAN package, so existing PMM users can adopt PATP without a new
mental model.}
\label{tab:estempmm-roadmap}
\end{table}

The graduation does not require any new theorems and keeps all assumptions
of Theorem~\ref{thm:g2-conditions} explicit at the API boundary.
Asymmetric distributions will trigger an automatic proxy-fallback warning
that quotes the residual bias from the bias term of
Algorithm~\ref{alg:patp-full}.


\end{document}